\documentclass[a4paper,11pt]{article}
\pdfoutput=1

\usepackage[utf8]{inputenc}
\usepackage[english]{babel}
\usepackage[T1]{fontenc}

\usepackage{amsmath}
\usepackage{amssymb}
\usepackage{amsthm}
\usepackage{amsfonts}
\usepackage{hyperref}
\usepackage{tikz}
\usepackage{tikz-cd}
\usepackage{extpfeil} 
\usepackage{authblk}
\usepackage{a4wide}
\usepackage{mathtools, nccmath, hhline}
\usepackage{makecell}
\setcellgapes{3pt}

\usepackage{arydshln}

\RequirePackage{subcaption}

\theoremstyle{definition}

\makeatletter
\newcommand{\verbatimfont}[1]{\renewcommand{\verbatim@font}{\ttfamily#1}}
\makeatother

\newcommand{\Eref}[1]{(\ref{#1})}
\newcommand{\Fref}[1]{Fig.~\ref{#1}}

\newcommand{\Rank}{\mathrm{rank}}

\newcommand{\Center}{\mathcal{Z}}

\newcommand{\GL}{\mathrm{GL}}

\newcommand{\Sp}{\mathrm{Sp}}
\newcommand{\Mp}{\mathrm{Mp}}

\newcommand{\ASp}{\mathrm{ASp}}
\newcommand{\Unitary}{\mathrm{U}}

\newcommand{\Cliff}{\mathrm{Cliff}}
\newcommand{\SCliff}{\mathrm{SemiCliff}}
\newcommand{\Pauli}{\mathrm{Pauli}}

\newcommand{\Integer}{\mathbb{Z}}
\newcommand{\ZZ}{\mathbb{Z}}
\newcommand{\QQ}{\mathbb Q}

\newcommand{\Field}{\mathbb{F}}

\newcommand{\tensor}{\otimes}

\usepackage{mathtools}

\newcounter{numitem}

\newcommand{\ket}[1]{|{#1}\rangle}

\newcommand{\bra}[1]{\langle{#1}|}

\newcommand{\nsp}{\!\!}

\renewcommand{\arraystretch}{1.2}

\newtheorem{lemma}{Lemma}[section]

\newtheorem{theorem}[lemma]{Theorem}

\newtheorem{claim}[lemma]{Claim}

\newcommand{\todo}[1]{ }

\newcommand{\igr}[1]{\includegraphics{images/#1.pdf}}
\newcommand{\igs}[1]{\includegraphics[scale=0.7]{images/#1.pdf}}
\newcommand{\igc}[1]{\raisebox{-0.45\height}{\includegraphics{images/#1.pdf}}}
\newcommand{\igd}[2]{\raisebox{#1\height}{\includegraphics{images/#2.pdf}}}
\newcommand{\igcs}[1]{\raisebox{-0.45\height}{\includegraphics[scale=0.7]{images/#1.pdf}}}
\newcommand{\hsp}{\ \ \ \ }
\newcommand{\Hsp}{\ \ \ \ \ \ \ \ }
\newcommand{\Double}{\mathfrak{D}}
\newcommand{\Parity}{\mathcal{H}}
\newcommand{\Symp}{\Field_2^n\oplus\Field_2^n}

\newcommand{\bS}{\bar{S}}
\newcommand{\bH}{\bar{H}}
\newcommand{\bX}{\bar{X}}
\newcommand{\bZ}{\bar{Z}}

\begin{document}

\title{Genons, Double Covers and Fault-tolerant Clifford Gates}

\author{Simon Burton, Elijah Durso-Sabina, Natalie C. Brown
}
\affil{\small Quantinuum}

\maketitle

\begin{abstract}
A great deal of work has been done developing quantum codes with varying overhead and connectivity constraints. 
However, given the such an abundance of codes,
there is a surprising shortage of fault-tolerant logical gates supported therein.
We define a construction, 
such that given an input $[[n,k,d]]$ code,
yields a $[[2n,2k,\ge d]]$ \emph{symplectic double} code
with naturally occurring fault-tolerant logical Clifford gates.
As applied to 2-dimensional
$D(\Integer_2)$-topological
codes with \emph{genons} (twists)
and domain walls, we find the symplectic double is genon free,
and of possibly higher genus. Braiding of genons on the original
code becomes Dehn twists on the symplectic double.
Such topological operations are particularly suited for
architectures with all-to-all connectivity, and 
we demonstrate this experimentally on Quantinuum's H1-1
trapped-ion quantum computer.
\end{abstract}

\tableofcontents

\section{Introduction}\label{sec:intro} 

Given that you have protected your fragile quantum
information from the world, how do you then gain
access to and thereby manipulate this quantum information?
This is the fundamental trade-off in the theory
of quantum error correction.
While the primary goal is to shield quantum data from
errors, the challenge lies in devising methods to interact
with and utilize this protected information for computational
tasks. Achieving this balance between protection and
accessibility is essential for realizing the full potential
of quantum error correction in practical applications.

One of the best techniques for circumventing this problem 
is to apply quantum gates 
\emph{transversally} between copies of the same
quantum code, see \cite{Gottesman1997} \S 5.3.
We view this copying process as analogous to 
the concept of \emph{covering spaces}.
The \emph{base code} $C$ is \emph{covered by} the \emph{total code} $C\oplus C$. 
(Here we are using additive $\oplus$ notation.)
Every qubit $j$ in $C$ is \emph{covered by} two qubits in $C\oplus C$.
These two qubits are called the \emph{fiber over} $j$.
Transversal gates 
operate separately on each fiber in the cover. 
We call these gates \emph{fiber transversal}. See \Fref{fig:double-cover}.

Making copies like this we only get trivial
covering spaces; a cartesian
product of a particular fiber $\mathbf{2}$ with the code $C$
giving $\mathbf{2}\times C = C \oplus C$. 
In this work we construct a 
double cover called the \emph{symplectic double} of $C$ 
and denoted $\Double(C)$. 
This cover is trivial when $C$ is a CSS code but becomes
non-trivial otherwise. 

A key idea is \emph{functoriality}:
logical Clifford operations on the base code 
can be lifted to logical operations on the total code.
In the symplectic double we often get an interesting set of fault-tolerant
Clifford gates. 
When the base Clifford gate is a product of single qubit Cliffords,
the lifted gate is fiber transversal, and 
these are fault-tolerant for the same
reason as in the trivial case: locally the cover is trivial.

Of particular interest to us is a familty of topological codes defined
on 3- and 4-valent graphs inscribed on compact surfaces
of arbitrary genus. We call these \emph{genon codes}.
The graph determines the code up to
local Clifford equivalence, with twists or \emph{genons} associated
to the 3-valent vertices.
This family of topological codes
includes surface codes, toric codes, hyperbolic surface codes,
$XZZX$ codes, the $[[5,1,3]]$ code.
See \Fref{fig:genon-codes}.

Applied to genon codes, the symplectic double corresponds to
a topological double cover, which is a discrete (combinatorial)
version of branched double covers of Riemann surfaces. 
Counting the genus of the base space and the double cover space
fixes the number of genons in the base, and this is the origin of
the term \emph{genon}. These are also called
twists, dislocations or defects in the literature 
\cite{Bombin2010,Barkeshli2013,Hastings2015,Brown2017,Yoder2017,Bombin2023}.

\begin{figure*}[t]
\centering
\begin{subfigure}[b]{0.3\textwidth}
$$
\begin{tikzcd}
F \arrow[r, rightarrowtail] & E \arrow[d, twoheadrightarrow] \\
            & B
\end{tikzcd}
$$
\caption{The \emph{fiber} $F$ is included into the \emph{total space} $E$ 
which projects down onto the \emph{base space} $B$.}
\end{subfigure}
~
\begin{subfigure}[b]{0.3\textwidth}
\includegraphics{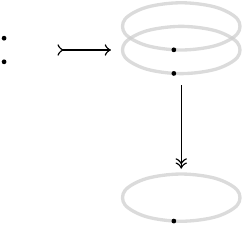}
\caption{A trivial double cover is just two copies of the base space, or
the cartesian product of the base with the fiber.}
\end{subfigure}
~
\begin{subfigure}[b]{0.3\textwidth}
\includegraphics{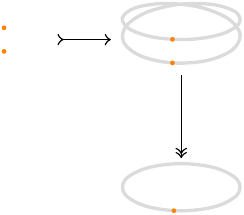}
\caption{A twisted double cover: locally this looks like a trivial
double cover, globally there is a twist.}
\end{subfigure}
\caption{
A schematic depiction of the symplectic double as
a covering space.
(ii) 
Given any CSS code $C$, 
the two copies $C\oplus C$
give a trivial double cover of $C$, and we have a logical 
transversal CX gate applied to the \emph{fibers} of the cover.
(iii)
When $C$ is not a CSS code we find there is a twisted
double cover $\Double(C)$
that also supports Clifford gates applied to the fibers.
}
\label{fig:double-cover}
\end{figure*}
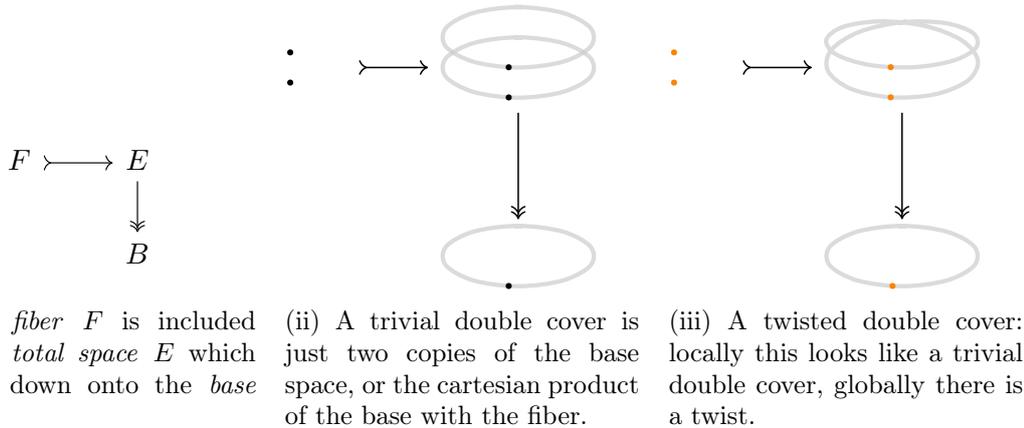

\begin{figure*}[]
\centering
\begin{subfigure}[t]{0.3\textwidth}
\igr{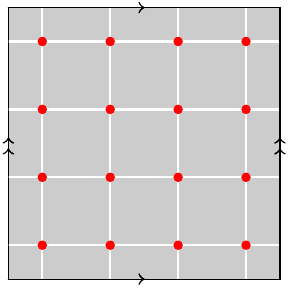} 
\caption{A graph $\Gamma$ on a torus with 
16 vertices in red, 32 edges in white and 16 faces in grey.}
\label{fig:toric-code-0}
\end{subfigure}
~
\begin{subfigure}[t]{0.3\textwidth}
\igr{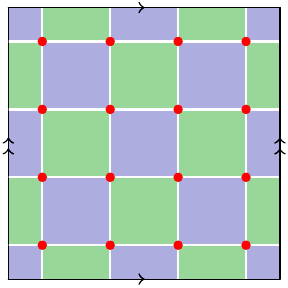} 
\caption{The faces of $\Gamma$ are bicolourable:
we paint each face with one of two colours blue or green
such that neighbouring faces have different colours.}
\label{fig:toric-code-1}
\end{subfigure}
~
\begin{subfigure}[t]{0.3\textwidth}
\igr{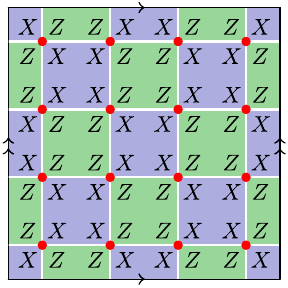}
\caption{
We associate a qubit with each vertex and a stabilizer with
each face. Using the bicolouring to determine $X$ or $Z$ type stabilizers
we recover the usual (rotated) toric code.}
\label{fig:toric-code-2}
\end{subfigure}
\begin{subfigure}[t]{0.3\textwidth}
\igr{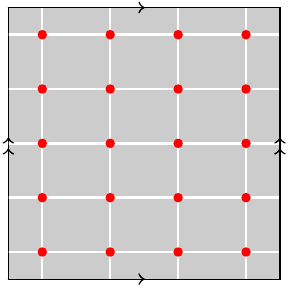} 
\caption{The graph $\Gamma$ has an odd number of faces
in the vertical direction which prohibits bicolouring the faces.}
\end{subfigure}
~
\begin{subfigure}[t]{0.3\textwidth}
\igr{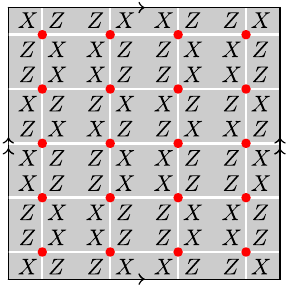} 
\caption{We can still define a quantum code as above 
but some stabilizers are both $X$ and $Z$ type and
this is a non-CSS code.}
\end{subfigure}
~
\begin{subfigure}[t]{0.3\textwidth}
\igr{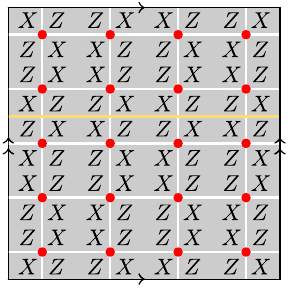}
\caption{We insert a \emph{domain wall}
separating or resolving the $X$ and $Z$ sectors appropriately.
}
\end{subfigure}
\begin{subfigure}[t]{0.3\textwidth}
\igr{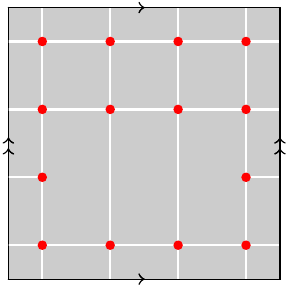} 
\caption{This graph $\Gamma$ has two trivalent vertices. }
\end{subfigure}
~
\begin{subfigure}[t]{0.3\textwidth}
\igr{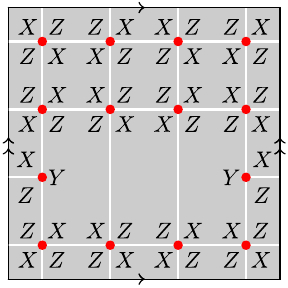} 
\caption{We place $XYZ$ around trivalent vertices . }
\end{subfigure}
~
\begin{subfigure}[t]{0.3\textwidth}
\igr{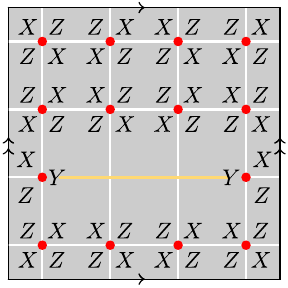}
\caption{Here the domain wall connects the two trivalent vertices.}
\end{subfigure}
\caption{Constructing a rotated toric code on a 
graph with bicolourable faces (i)-(iii).
The frustrated bicolourability of $\Gamma$ is
related to domain walls (iv)-(vi).
Trivalent vertices, or \emph{genons},
are another reason bicolourability is frustrated (vii)-(ix).}
\label{fig:genon-codes}
\end{figure*}

\subsection{Outline}

We discuss the underlying theory of two-dimensional 
topological order
with $D(\Integer_2)$ anyon excitations in \S\ref{sec:topo}.
This calculus is divorced from the microscopic details of
the system, so we don't need to take into account any qubit or
stabilizer structure, and can instead focus on the large-scale topological 
properties of the system. 
Crucially this phase has one non-trivial symmetry which we denote
using one-dimensional domain walls.
The endpoints of domain walls are the genons, and
by braiding these we can perform logical Clifford gates.

The symmetry exhibited by domain walls has another
incarnation as precisely 
the data needed to construct a \emph{double cover}
of our two-dimensional manifold, or \emph{Riemann surface}
\S\ref{sec:riemann}.
Topological operations in the base space, such as braiding
genons, will then
lift  functorially to topological operations in the total space,
such as Dehn twists.

We review background on quantum stabilizer codes and notation in \S\ref{sec:qcodes}
and then introduce the \emph{symplectic double} construction in \S\ref{sec:sp},
and show how logical Clifford gates on the base code 
lift functorially 
to logical Clifford gates on the total code \S\ref{sec:func},
as well as their fault-tolerance.
These lifted Clifford gates lie in the phase-free $ZX$-calculus.

We introduce our formalism for topological codes with genons in
\S\ref{sec:graphs}. These are called \emph{genon codes},
and come with a theory of logical string operators which is invariant under
local Clifford operations \S\ref{sec:string}, as well as
rules for decoration by \emph{domain walls}\S\ref{sec:domain}.
When a genon code has no unnecessary $Y$ operators
the symplectic double will also be a genon code \S\ref{sec:double}.

In \S\ref{sec:examples} we go through a series of example
genon codes and symplectic doubles.
These have interesting Clifford gates, 
which we call \emph{genon protocols}. 
The smallest interesting example is a $[[4,1,2]]$ genon code 
with 4 genons, whose symplectic double is a $[[8,2,2]]$ toric code.
We show how braiding these four genons gives a protocol for
implementing the single qubit Clifford group, as well as
Dehn twists on the $[[8,2,2]]$ toric code.

We experimentally demonstrate three different protocols 
on Quantinuum's H1-1 trapped-ion quantum computer, including
genon braiding on the $[[4,2,2]]$ code, Dehn
twists on the $[[8,2,2]]$ code and a lifted Clifford gate on the
$[[10,2,3]]$ code \S\ref{sec:experiments}.
Such experiments are a ``proof-of-principle,''
demonstrating that
the gates arising from the procedures can be realized on modern quantum computers.
We conclude in \S\ref{sec:conc}.

We aim to use a consistent colour scheme as an 
aid to recognizing
the same concept as it appears in seemingly different guises.
Qubits are red dots, $X$-type string operators are blue,
and $Z$-type string operators are green. 
This blue/green
colour scheme also applies to $ZX$-calculus.
The symmetry that swaps blue and green is coloured yellow
and occurs both as domain walls and the Hadamard gate.

Finally, there are several \emph{Claims} in this work which can be taken to
be either unproven theorems, or conjectures.

\subsection{Related work}

This work necessarily treads a well-worn path, and 
we cite some of the work to be found on this path here.

\begin{itemize}
\item
This work is a continuation of
the first author's work on $ZX$-dualities and folding 
\cite{Breuckmann2022}:
``fold-transversal'' logical gates are 
examples of what is here called fiber transversal. 

\item
The symplectic double construction also appears in \cite{Liu2023},
where it is applied to subsystem codes.

\item
A particular example of a genon code without any genons is
the $XZZX$ code \cite{Bonilla2021}, and the $[[5,1,3]]$ code.

\item
Surface codes are an 
example of genon codes on a sphere:
in \cite{Gowda2020} they consider modifying such face-bicolourable 
4-valent graphs to support 3-valent dislocations,
as well as qudit generalizations.

\item
Another treatment of surface codes and twists
is found in \cite{Bombin2023}, where the focus is on
2+1 dimension circuit construction.

\item
The family of genon codes defined in this work 
overlaps and is inspired by the graph-based formalism of 
Sarkar-Yoder \cite{Sarkar2021}.
They were the first to identify the well-known
$[[5,1,3]]$ as a member of a family of topological codes defined on a torus.
In the Sarkar-Yoder formalism they take genons (branch points)
to occur precisely at the trivalent (or odd-valent)
vertices, and so
obtain a deficiency in the vertex count of the double cover (see proof of Thm. 4.9).
In this work we take genons (branch points) to occur on faces
near to the trivalent vertices, giving exactly twice the number
of vertices in the double cover (and a deficiency in the doubled faces).
This works better for our purposes as we then find a connection with the
symplectic double, from which our other code constructions follow.
The Sarkar-Yoder formalism takes the branch cuts between genons
to consist of paths along edges of the graph with endpoints at the
trivalent vertices. These paths are called defect lines
and they show how these relate to a bicolour (or checkerboard)
on the faces.
In this work we take branch cuts to be paths perpendicular to the edges,
which then give the \emph{domain walls} for the underlying topological order.
The Sarkar-Yoder formalism also contains a doubling construction in \S 4.3,
that reproduces our symplectic double when there are no genons present.
For example, the Sarkar-Yoder double of the $[[5,1,3]]$ code is the $[[10,2,3]]$
toric code.

\item
A description of Dehn twists on topological
codes appears in \cite{Breuckmann2017}, \S 4.2.

\item
In \cite{Lavasani2019} Appendix D, they
describe a protocol for instantaneous Dehn twists involving
permutations of qubits followed by 
Pachner moves.
These Pachner moves are implemented
using constant depth CX circuits and are
designed to re-triangulate the code lattice after performing
the qubit permutation.
Their work is developed further in \cite{Zhu2020} where 
more protocols for performing braids of defects (genons) and
Dehn twists are given.

\item
Further discussion on Clifford gates on topological codes
is in \cite{Brown2017,Lavasani2018}.
\end{itemize}

\section{Background}\label{sec:background}

\subsection{$D(\Integer_2)$ topological order}\label{sec:topo} 

\begin{figure*}[!ht]
\centering
\begin{subfigure}[t]{0.24\textwidth} \igs{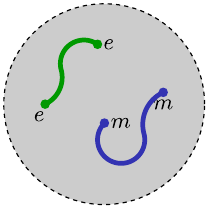}
\caption{Anyons at the endpoint of logical operators.}
\end{subfigure} 
\begin{subfigure}[t]{0.24\textwidth} \igs{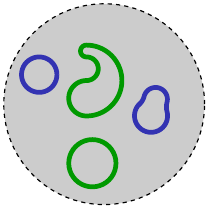} 
\caption{Contractible loops have no effect on the encoded state.}
\end{subfigure} 
\begin{subfigure}[t]{0.24\textwidth} \igs{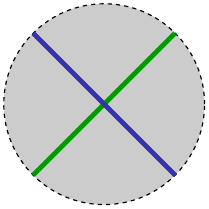} 
\caption{Every green-blue crossing introduces a factor of $-1$
in the commutation relation.}
\end{subfigure} 
\begin{subfigure}[t]{0.24\textwidth} \igs{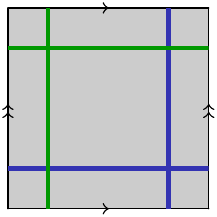} 
\caption{On a torus we have two pairs of anti-commuting logical operators,
encoding two qubits.}
\end{subfigure} 
\begin{subfigure}[t]{0.24\textwidth} \igs{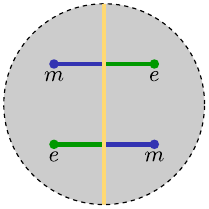}
\caption{The $H$ domain wall exchanges $e$ and $m$}
\end{subfigure} 
\begin{subfigure}[t]{0.24\textwidth} \igs{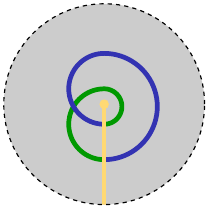}
\caption{Contractible loop around a genon.}
\end{subfigure} 
\begin{subfigure}[t]{0.24\textwidth} \igs{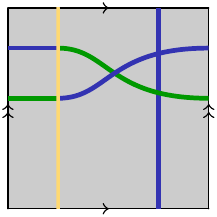}
\caption{A non-contractible domain wall on a torus; this system encodes
a single qubit.}
\end{subfigure} 
\\
\begin{subfigure}[t]{0.24\textwidth} \igs{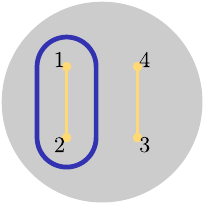}
\caption{Logical $\bar{X}$ }
\end{subfigure} 
\begin{subfigure}[t]{0.24\textwidth} \igs{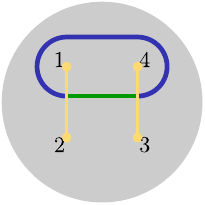}
\caption{Logical $\bar{Z}$ }
\end{subfigure} 
\begin{subfigure}[t]{0.24\textwidth} \igs{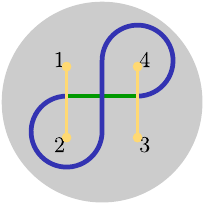}
\caption{Logical $\bar{Y}$ }
\end{subfigure} 
\caption{String operators are coloured blue and green, 
which are $X$- and $Z$-type respectively.
The $H$ domain wall is a yellow string, whose (any) endpoints are called genons.
(viii)-(x): A sphere supporting four genons encodes one logical qubit.
}
\label{fig:strings}
\end{figure*}

In this section, we briefly discuss the theory of $D(\Integer_2)$
topological order, describing anyon statistics through
the use domain walls, defects and genons.
For a more in depth and accessible introduction we recommend
\cite{Brown2017}, as well as
\cite{Barkeshli2019} \S VI and \cite{Carqueville2023}.

To start, consider an abelian anyon model with four anyon types:
\emph{vacuum}, \emph{electric}, \emph{magnetic} and \emph{electromagnetic} anyons.
We denote these $\iota, e, m, \varepsilon$ respectively (See \Fref{fig:strings}).
Fusion rules describe how pairs of anyons combine to
form other anyons. 
These rules are as follows:
\begin{align*}
\iota \times a &= a \times \iota = a \quad \text{for all anyon labels } a \\
e \times e &= 
m \times m = \varepsilon \times \varepsilon = \iota\\
e \times m &= m\times e = \varepsilon \\
e \times \varepsilon &= \varepsilon\times e =  m \\
m \times \varepsilon &=  \varepsilon\times m = e \\
\end{align*}

We can describe the path of an anyon on a topological
surface by a string operator whose end points are the anyons.
When these string operators are closed, they are the 
logical operators of the topological code.
Here we adopt the convention that blue strings connecting
two $e$ anyons are $X$-type operators and green strings
connecting two $m$ anyons are $Z$-type operators.
These strings can cross and annihilate ends points
following the fusion rules (see. Fig. ~\ref{fig:strings}).

This topological order supports two automorphisms: the
trivial automorphism and the automorphism that swaps
$e$ and $m$. These automorphisms occur when anyons cross
a boundary, or \emph{domain wall}, that separates regions. 
The trivial automorphism
is realized by anyons crossing a trivial domain wall,
while the $e-m$ swapping automorphism occurs when anyons
cross a non-trivial domain wall.
The endpoints of these non-trivial domain walls we call \emph{genons}.

Because $m\times m=\iota$ and $e\times e=\iota$
the associated string operators are self-inverse.
In other words, two copies of the same colour string
can be erased:
$$
\igcs{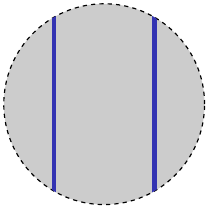}
\cong
\igcs{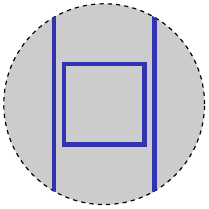}
\cong
\igcs{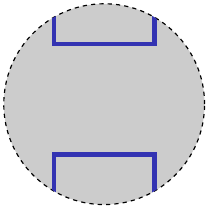}
$$
and similarly for $Z$-type string operators.

The $H$ domain wall exchanges $e \leftrightarrow m$
and so is also self-inverse but for a different reason:
$$
\igcs{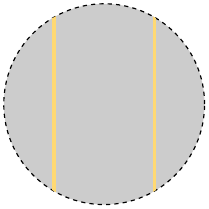}
\cong
\igcs{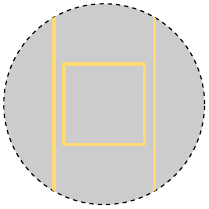}
\cong
\igcs{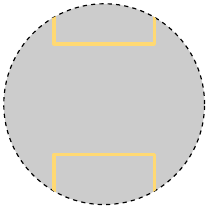}
$$

Contractible loops around genons
$$
\igcs{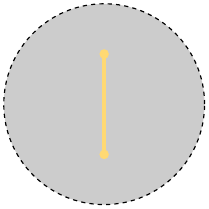}
\cong
\igcs{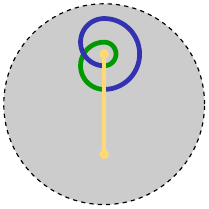}
\cong
\igcs{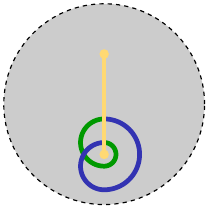}
$$
imply the following equations
\begin{align*}
&\igcs{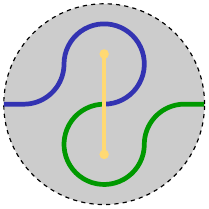}
\cong
\igcs{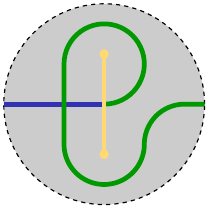}
\cong
\igcs{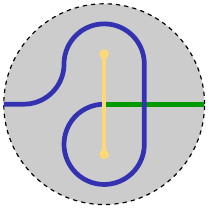} \\
\cong
&\igcs{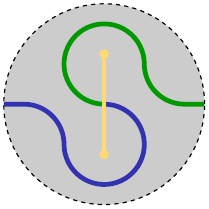}
\cong
\igcs{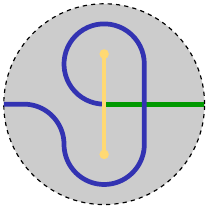}
\cong
\igcs{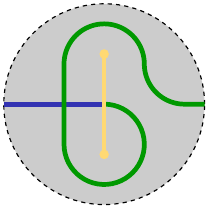}
\end{align*}
We see from this that performing a \emph{braid} between two genons
has the same result clockwise versus anti-clockwise:
\begin{align*}
\igcs{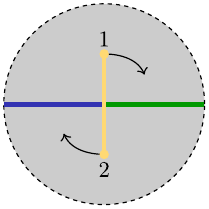}\mapsto
&\igcs{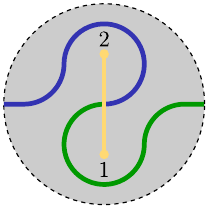}
\\
& \ \ \Hsp\shortparallel
\\
\igcs{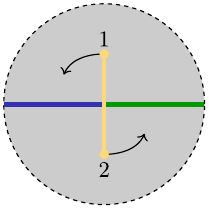}\mapsto
&\igcs{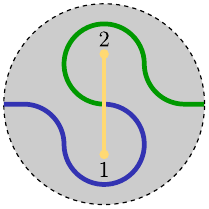}
\end{align*}
Therefore we can refer to a braid of genons by the
underlying permutation of the genon labels.

Connecting these concepts to the language of the
stabilizer formalism (see \S\ref{sec:qcodes}),
we record the following dictionary:
\begin{center}
\begin{tabular}{cc}
\underline{$D(\Integer_2)$ topological order} & \underline{Quantum stabilizer codes} \\
$X/Z$-type string operator & $X/Z$-type logical operator \\
contractible string operator  & stabilizer \\
$e/m$ anyon  & frustrated $X/Z$-type stabilizer  \\
\end{tabular}
\end{center}

\subsubsection{Example in genus zero}\label{sec:genus-zero} 

So far we have only discussed \emph{local} rules for
anyon calculations, as depicted by the dashed line
surrounding grey regions. In this section we consider
the effect that the \emph{global} topology has.

A sphere with four genons (two domain walls), encodes
one logical qubit, \Fref{fig:strings}.
By topological deformation, we braid these genons.
Here we number the genons $1,2,3,4$ and see the action
on the logical operators as we perform these braids.

The first gate we try is given by the
permutation operator $\sigma = (1,4,3,2)$.
This is swapping genons 2 and 4, and we show a clockwise braid of these two
genons:
\begin{align*}
\bar{X} = \igcs{topo_1234_X}
&\mapsto
\igcs{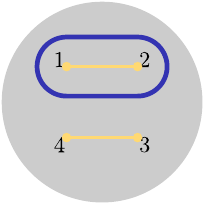}
\cong
\igcs{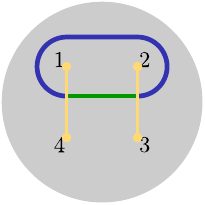} = \bar{Z},
\\
\bar{Z} = \igcs{topo_1234_Z} 
&\mapsto
\igcs{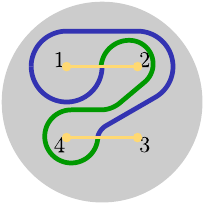}
\cong
\igcs{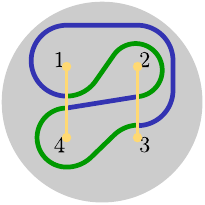}
\\
&\cong
\igcs{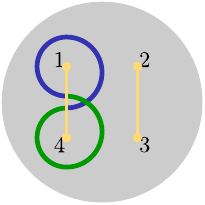}
\cong
\igcs{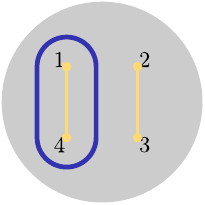}
= \bar{X}.
\end{align*}
And therefore this braid implements a logical $H$ gate.

The next permutation we examine is $\sigma = (1,3,2,4)$
which swaps genons 2 and 3.
\begin{align*}
\bar{X} = \igcs{topo_1234_X}
&\mapsto
\igcs{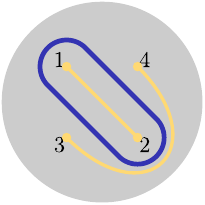}
\cong
\igcs{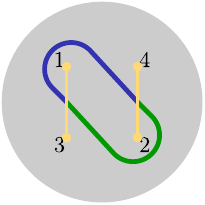}
\cong \bar{Y},
\\
\bar{Z} = \igcs{topo_1234_Z} 
&\mapsto
\igcs{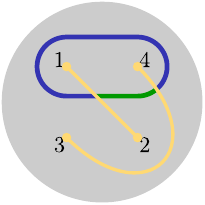}
\cong
\igcs{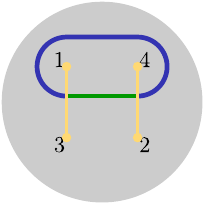}
=\bar{Z}.
\end{align*}
And this gives a logical $S$ gate.

Finally we note that the two permutations
$(2,1,4,3)$ and $(3,4,1,2)$ leave the logical
operators invariant.
This is because we are working on a sphere and can
drag operators around the back of the sphere:
\begin{align*}
\bar{X} = \igcs{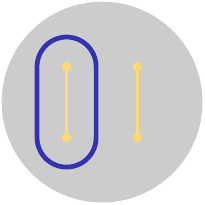} &\cong
\igcs{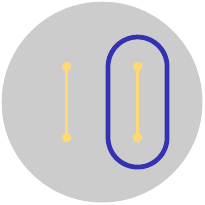}, \\
\bar{Z} = \igcs{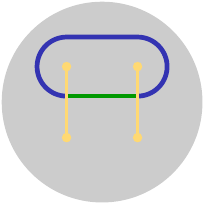} &\cong 
\igcs{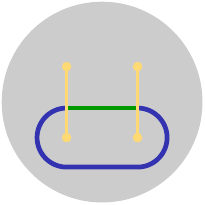} \cong
\igcs{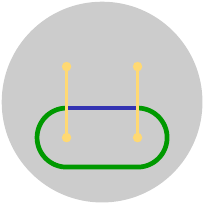}.
\end{align*}

We will show in \S\ref{sec:examples} below a minimal implementation of
these gates in a four qubit quantum code.

\subsubsection{Dehn twists}

\begin{figure*}[t]
\centering
\begin{subfigure}[t]{0.4\textwidth}
\includegraphics[scale=0.1]{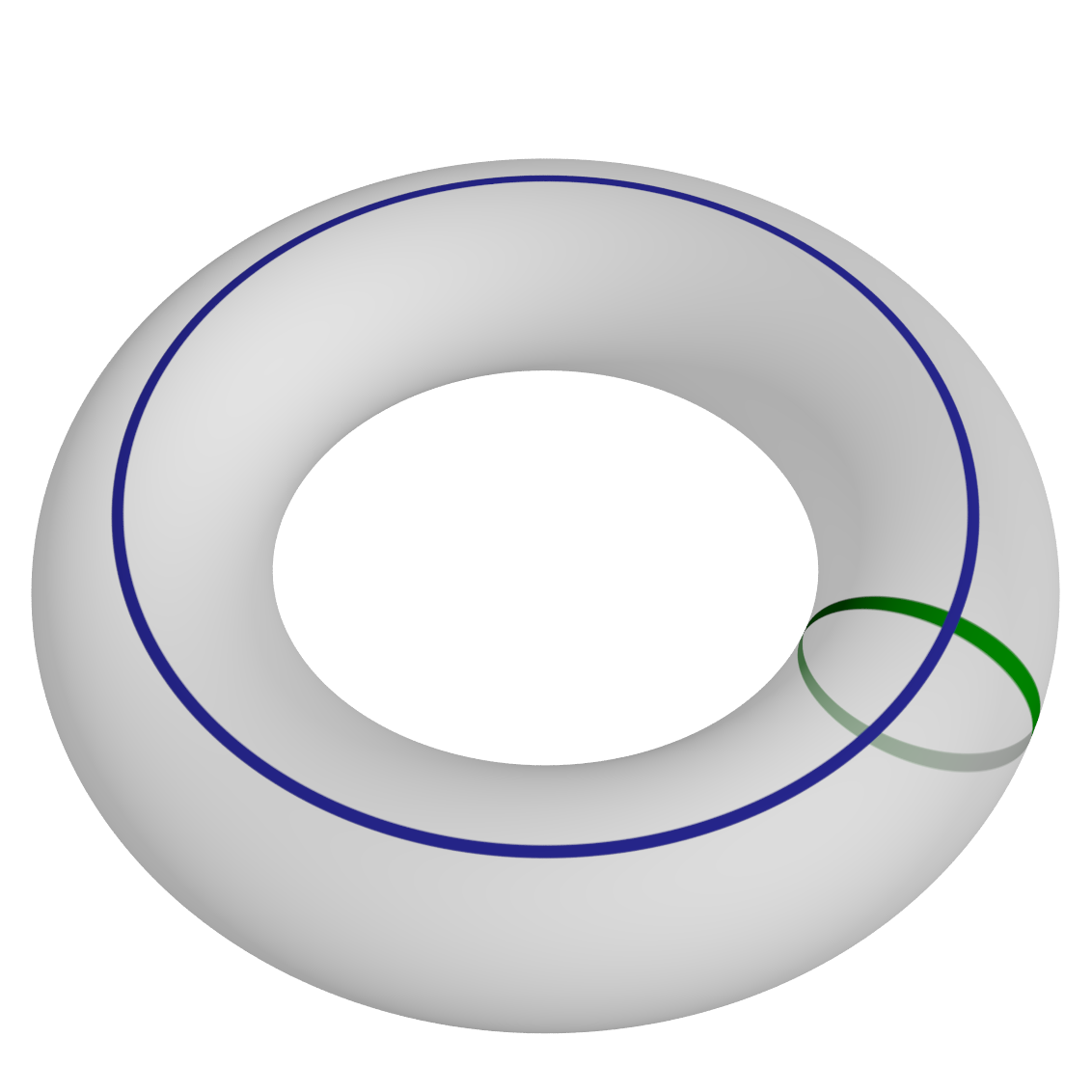}
\caption{The torus with an anti-commuting pair of logical operators.}
\end{subfigure}
~
\begin{subfigure}[t]{0.4\textwidth}
\includegraphics[scale=0.1]{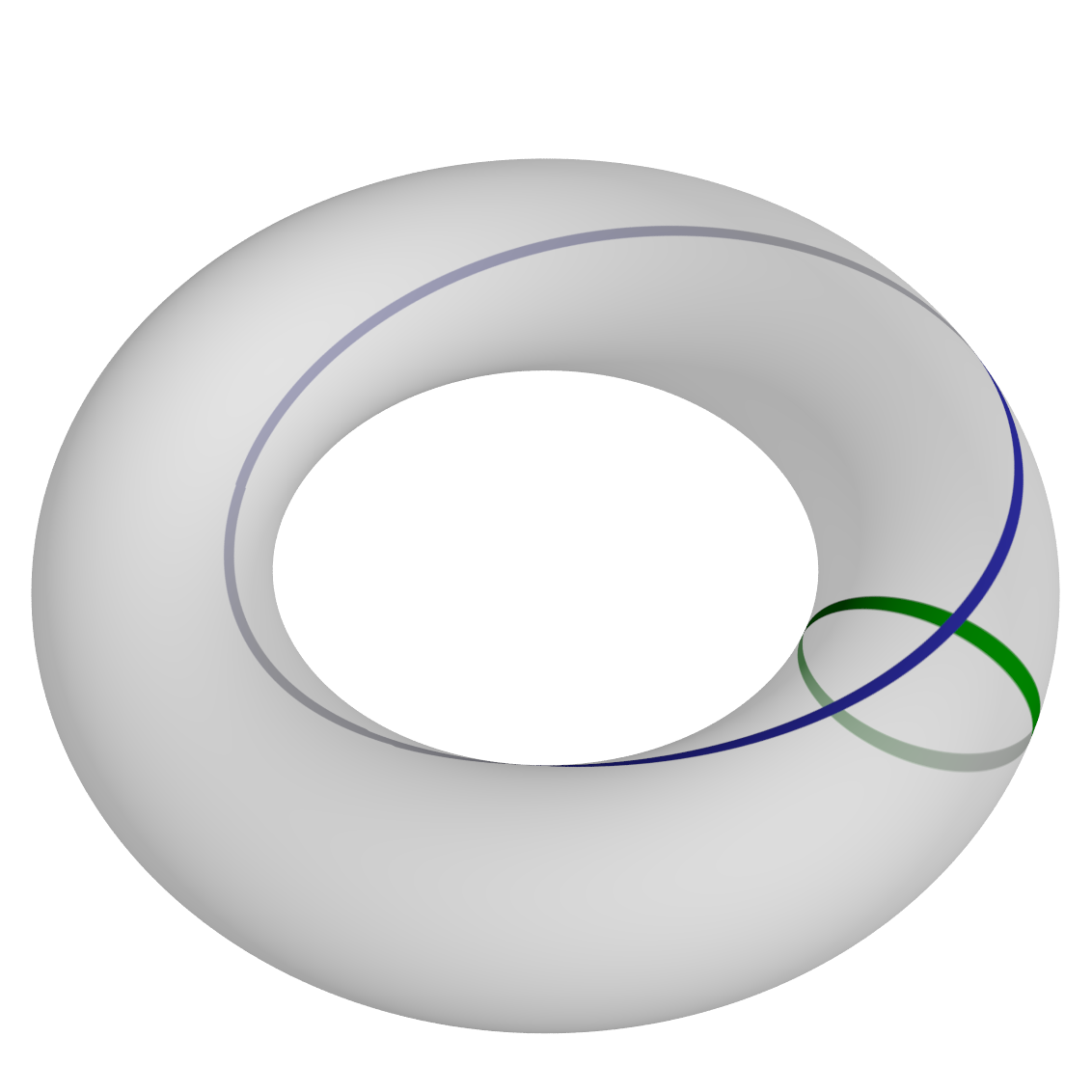}
\caption{A Dehn twist introduces a full rotation in the torus.
Here we see the blue string operator now winds around the
back of the torus.}
\end{subfigure}
\caption{A Dehn twist on a torus.}
\label{fig:dehn}
\end{figure*}

On a genus zero surface (sphere) the only non-trivial
homeomorphisms up to isotopy are the genon braids.
On a higher genus surface we have many more non-trivial
homeomorphisms.
A \emph{Dehn twist} on a torus $T$ is a homeomorphism 
of the torus that introduces a global twist in the torus,
see \Fref{fig:dehn}.
A horizontal Dehn twist is implemented as a linear 
shear operation, up to 
periodic boundary conditions:
\begin{align*}
\igc{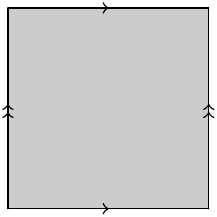} &\hsp\mapsto\hsp \igc{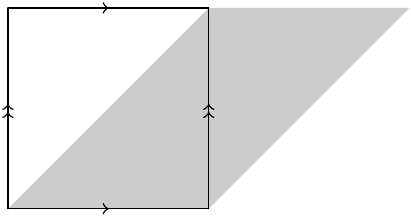}   \\
\end{align*}
Now we compute the action of this horizontal Dehn twist on 
the logical operators.
A complete set of logical operators is given in anti-commuting 
pairs $(\bX_0,\bZ_0)$ and $(\bX_1,\bZ_1)$:
\begin{center}
\igc{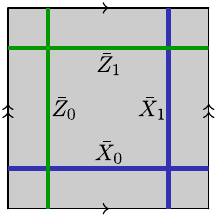}
\end{center}
The action of a horizontal Dehn twist is then found to be 
\begin{align*}
\igc{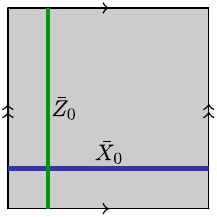} \ \ &\mapsto\ \ \igc{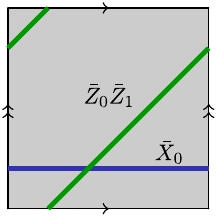} \\
(\bar{X}_0,\bar{Z}_0)\hsp\hsp\hsp &\mapsto \hsp\hsp(\bar{X}_0, \bar{Z}_0\bar{Z}_1) \\
\end{align*}
\begin{align*}
\igc{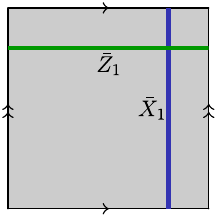} \ \ &\mapsto\ \ \igc{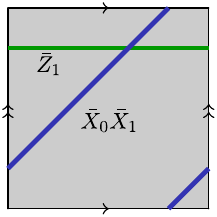} \\
(\bar{X}_1,\bar{Z}_1)\hsp\hsp\hsp &\mapsto \hsp\hsp(\bar{X}_0\bar{X}_1, \bar{Z}_1) 
\end{align*}
and therefore this Dehn twist implements logical $CX_{1,0}$
with control qubit $1$ and target qubit $0$.
A similar calculation shows that the vertical
Dehn twist 
\begin{align*}
\igc{topo_genus_1} &\hsp\mapsto \hsp
\raisebox{-0.24\height}{\includegraphics{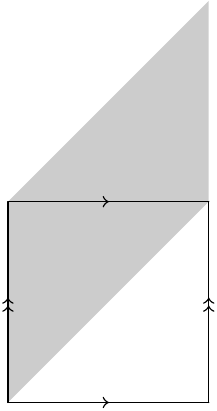}}
\end{align*}
implements a logical $CX_{0,1}$.
The combination of these two logical gates generates
the group $\GL(2,\Integer_2)$ which is isomorphic
to the permutation group $S_3$ of order 6.

These two Dehn twists are known to implement the 
\emph{mapping class group} of the torus~\cite{Farb2011}.
The mapping class group of a surface
is the group of isotopy classes of homeomorphisms of the surface.

\subsubsection{Riemann surfaces and double covers}\label{sec:riemann} 

In this section we re-interpret the above 
theory of $D(\Integer_2)$ topological order 
using the language of Riemann surfaces. 
Loosely speaking, 
if domain-walls correspond to passing between
the normal world and a ``bizarro'' world, then
why don't we interpret this literally?
In other words, take two copies of the topological
phase and cut/glue them together appropriately, along the domain walls.
This motivates the following consideration of 
\emph{branched double covers}.

Topologically, there are only two ways to double cover a circle,
which is the only compact connected one-dimensional manifold, see \Fref{fig:double-cover}.
When we do this with compact surfaces things get much more
interesting.
See the textbook \cite{Girondo2012} \S 1.2.5.

Compact oriented connected topological surfaces are characterized
by their \emph{genus}, which counts the number of ``holes'':
\begin{center}
\begin{tabular}{ccc}
  \includegraphics[scale=0.1]{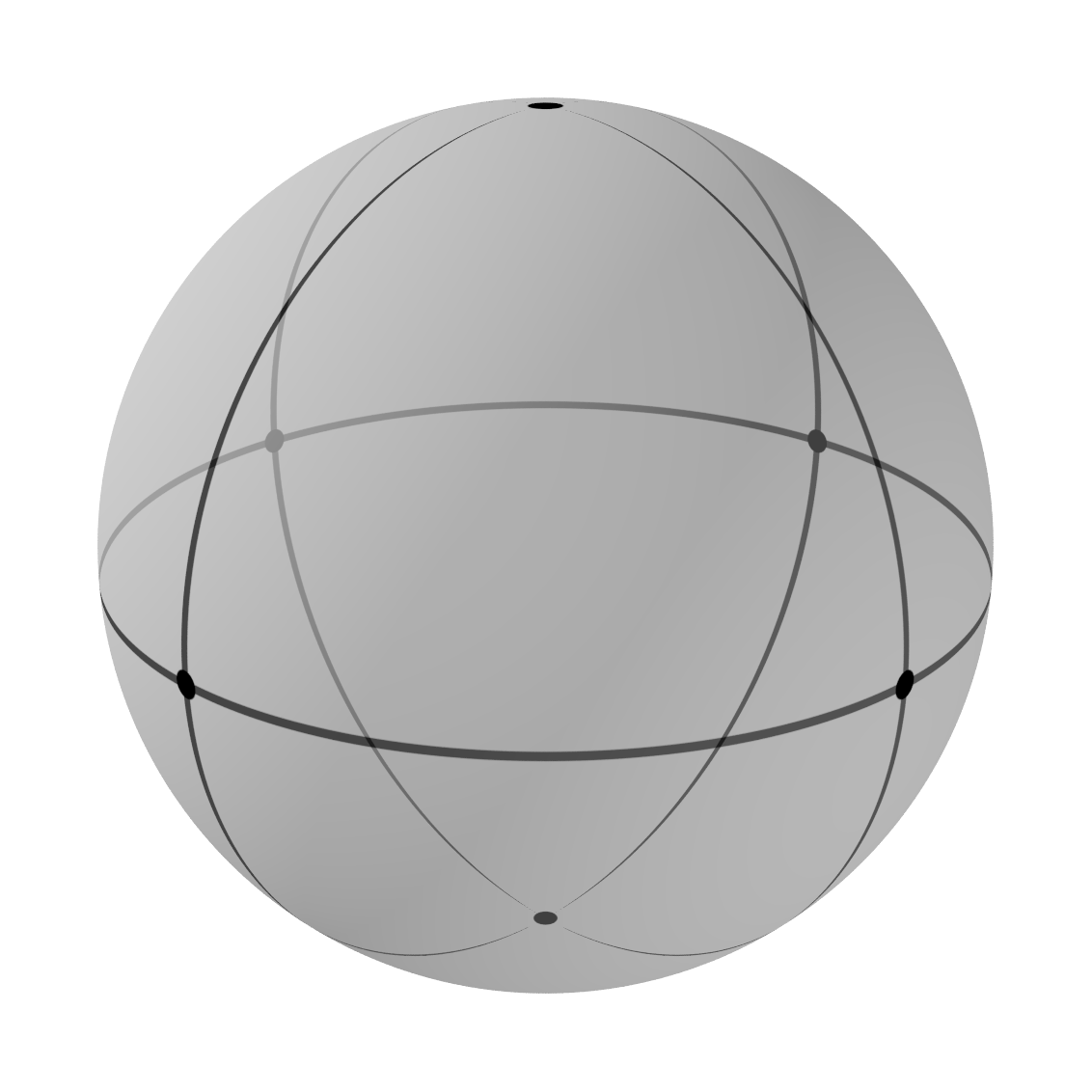}
& \includegraphics[scale=0.1]{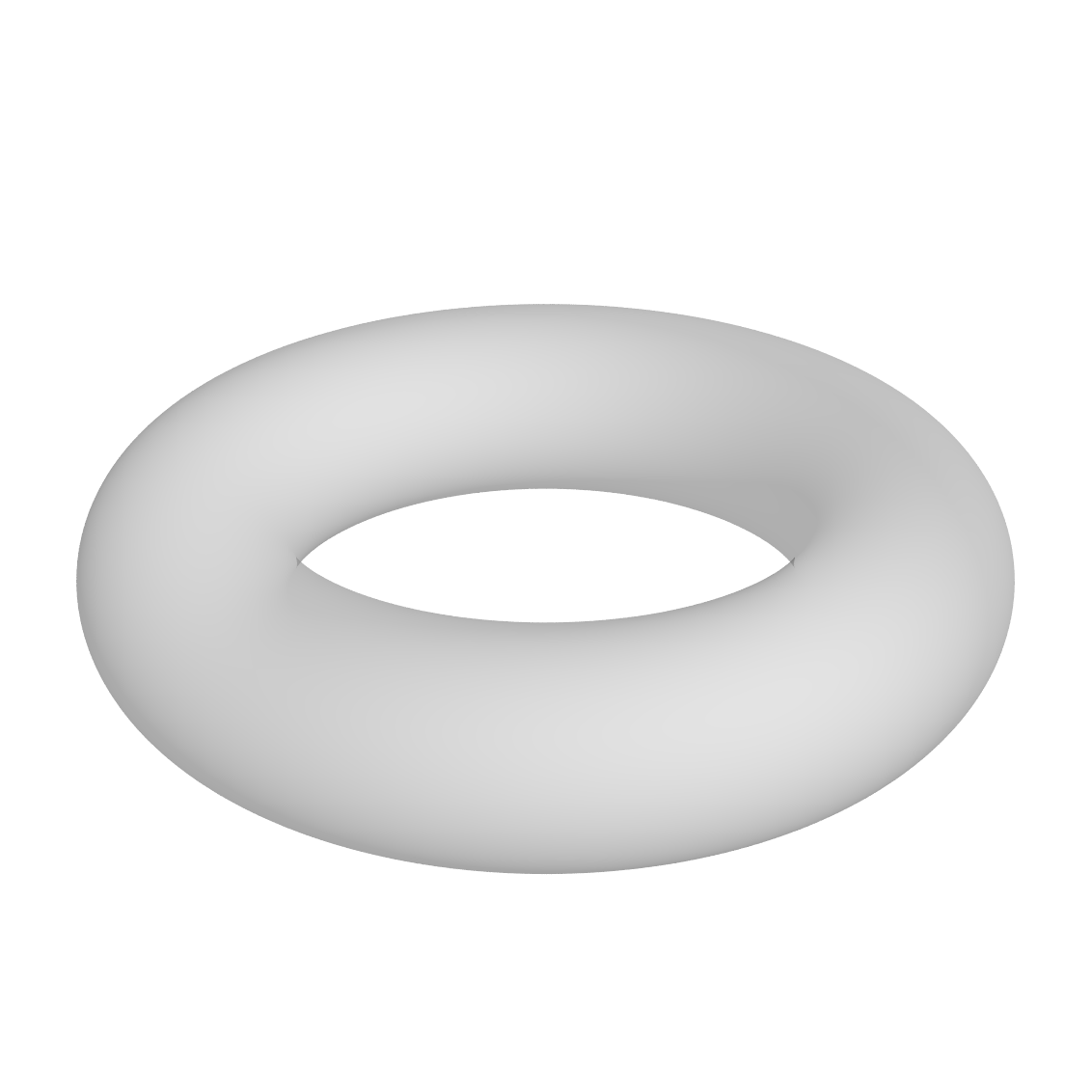}
& \includegraphics[scale=0.2]{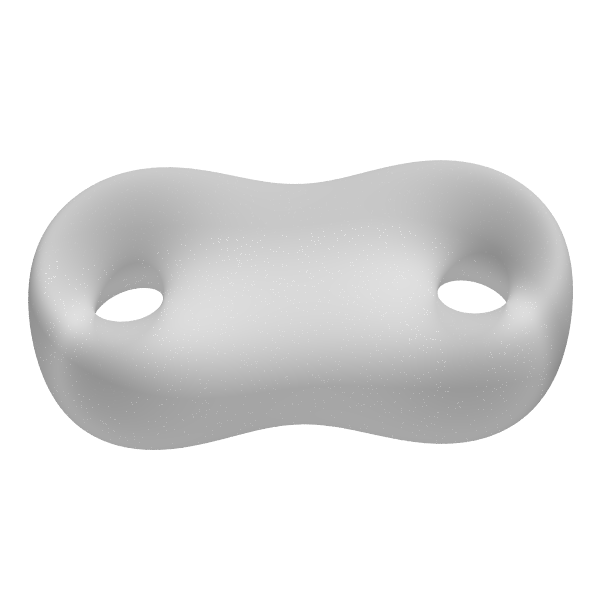}
\\
genus zero & genus one & genus two 
\end{tabular}
\end{center}
The genus zero surface is a sphere, and we have inscribed on it
$v=6$ vertices, $e=12$ edges joining vertices, and $f=8$ faces, thus
giving the \emph{Euler character}
$$
\chi = v - e + f = 6 - 12 + 8 = 2.
$$
In general, for a genus $g$ surface we have that
$\chi = 2 - 2g.$
The Euler character of a surface constrains what
graphs we can inscribe on that surface.

Given a compact surface $E$ that double covers a surface $B$,
$$
\begin{tikzcd}
E \arrow[d, twoheadrightarrow, "p"] \\
B
\end{tikzcd}
$$
the Euler characteristic $\chi(.)$ satisfies the formula
$$
    \chi(E) = 2\chi(B)
$$
If the cardinality of the fiber $p^{-1}(b)$ is $1$ at
a point $b\in B$ we say that the cover has a 
\emph{branch point} at $b$.
Such branch points introduce a correction into the
formula for the Euler characteristic:
$$
    \chi(E) = 2\chi(B) + m
$$
where $m$ is the number of branch points.
This is a special case of the more general 
Riemann-Hurwitz formula, see \cite{Girondo2012} Thm 1.76.
In terms of the \emph{genus} $g(.)$ of the surfaces $E$ and $B$,
we have
\begin{align*}
2g(E) - 2 &= 2(2g(B) - 2) + m, \\
g(E) - 1 &= 2g(B) - 2 + \frac{1}{2}m.
\end{align*}
When the genus of both $E$ and $B$ is zero
we find that $m=2$. This cover is given
by the double cover of the Riemann sphere,
or extended complex plane, under the function $f(z)=z^2.$
Most points, such as $1=(\pm 1)^2$ and $-1=(\pm i)^2$ are
double covered by $f$, except for the $m=2$ points at $0$ and $\infty$:
\begin{center}
\raisebox{-0.4\height}{\includegraphics[scale=0.1]{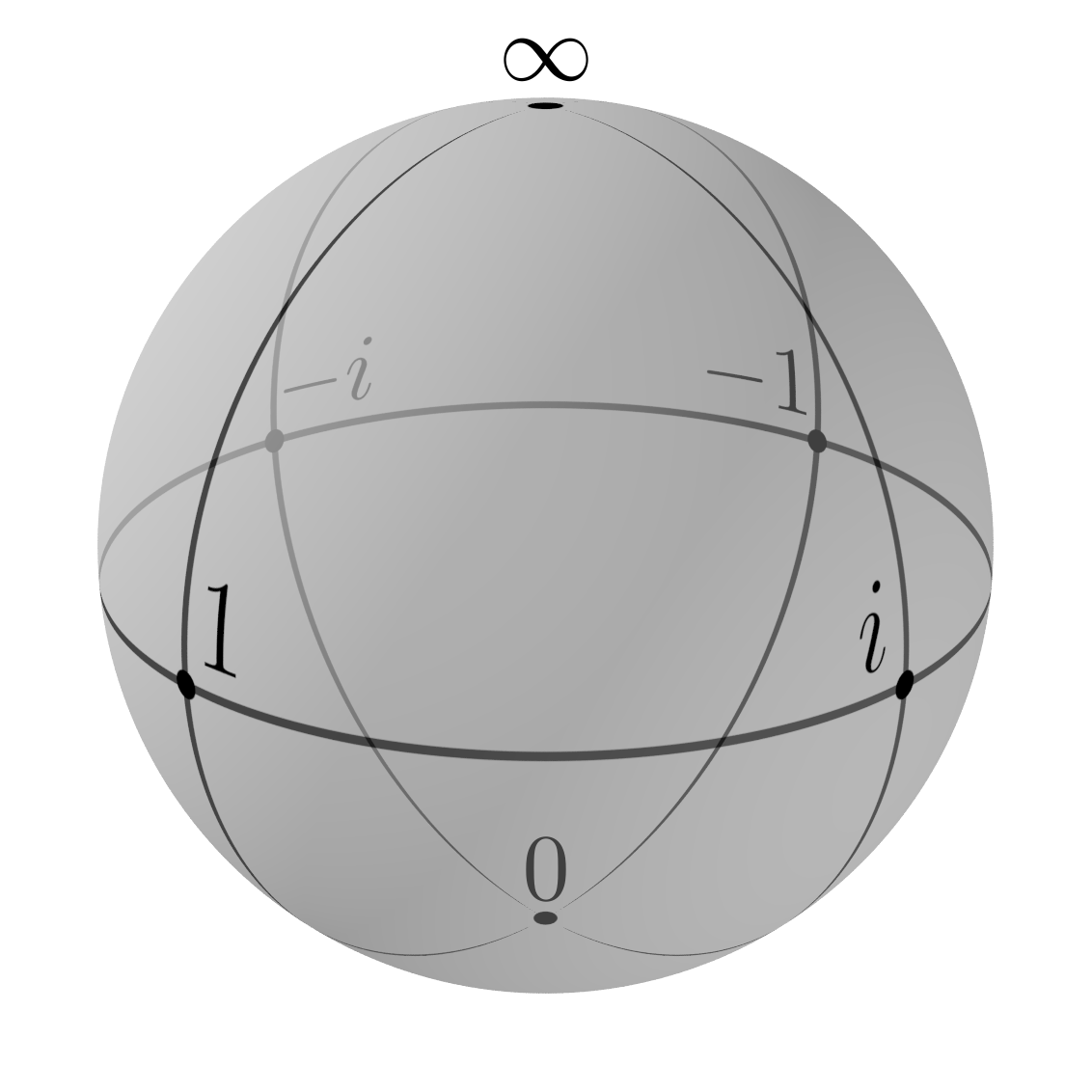}}
\hsp$\xtwoheadrightarrow{p}$\hsp
\raisebox{-0.4\height}{\includegraphics[scale=0.1]{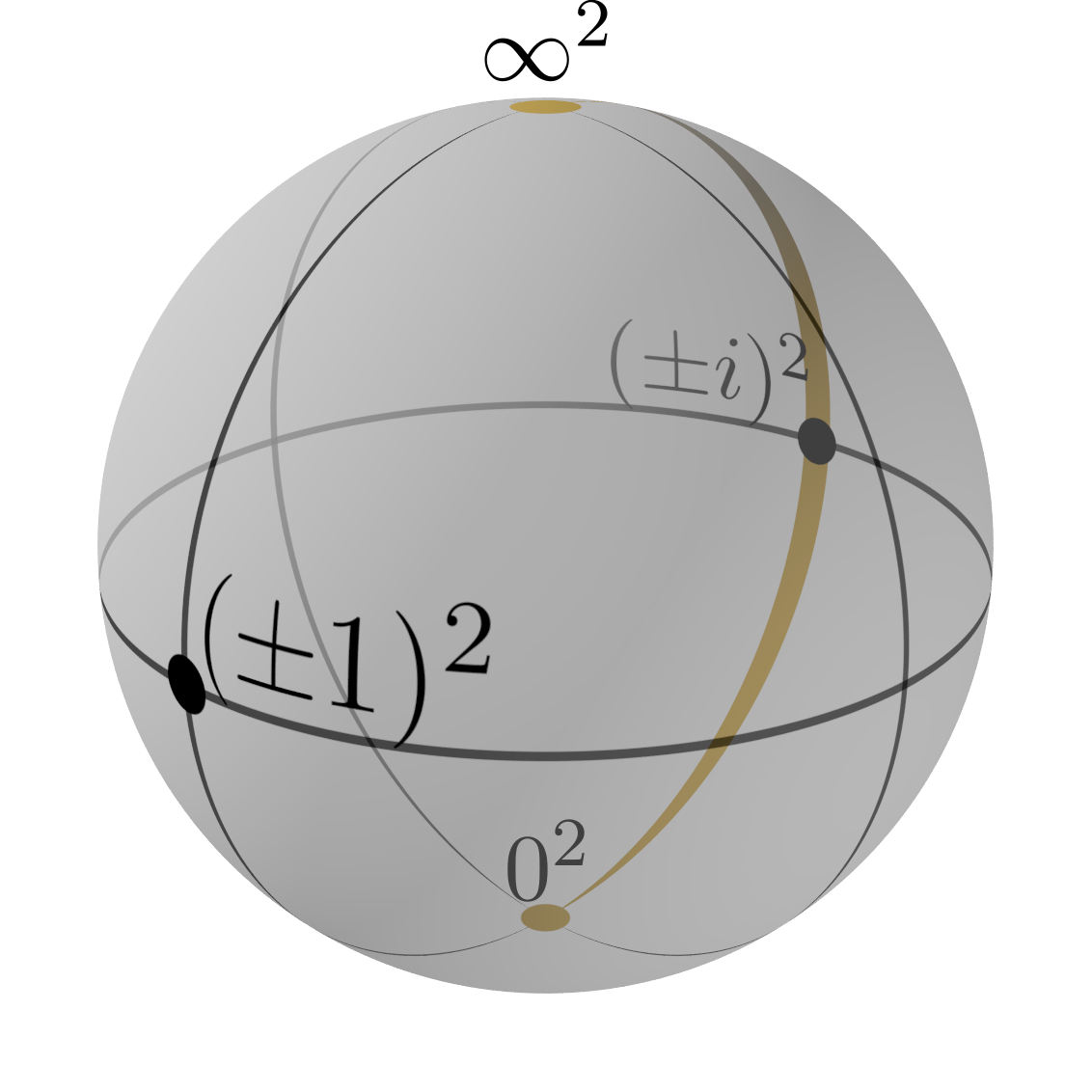}}
\end{center}
The yellow line is an arbitrary line connecting the two yellow branch points
at $z=0$ and $z=\infty$.
We can construct these branched covers with some cutting and glueing.
First we take two copies of the base space:
\begin{center} 
\includegraphics[scale=0.25]{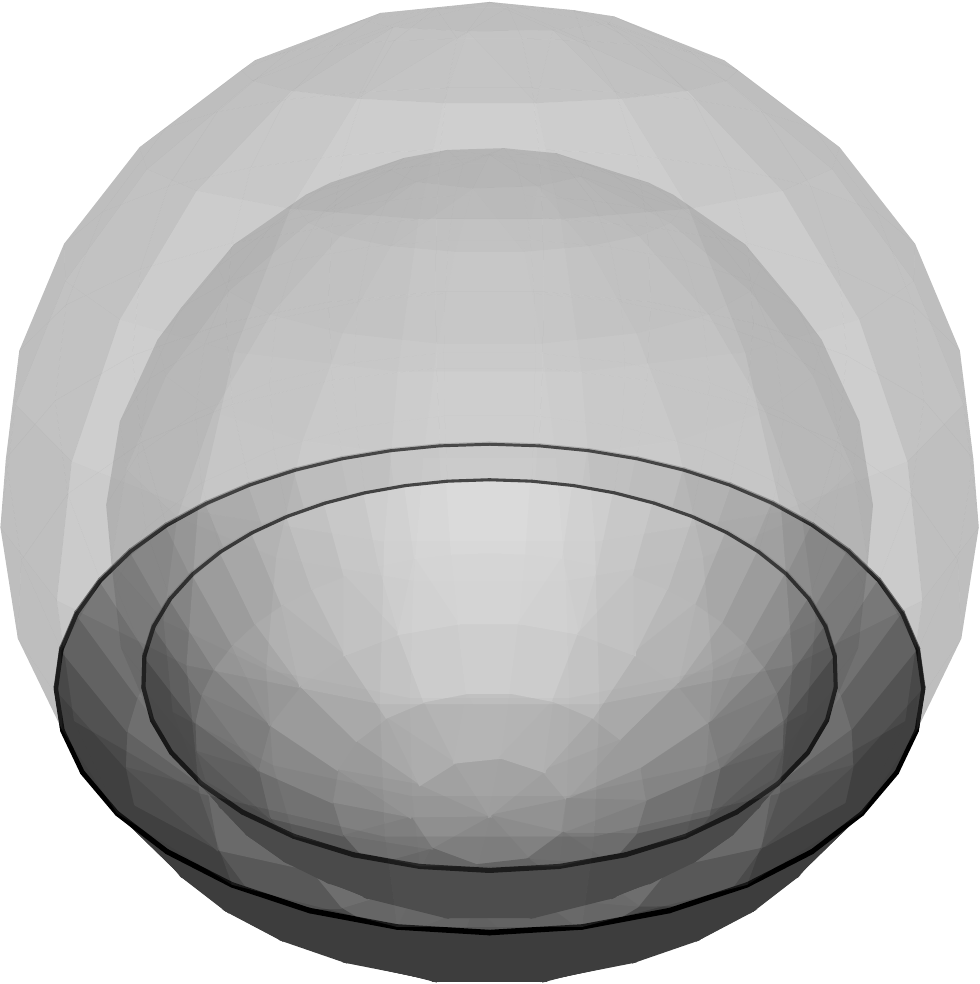}
\end{center}
This is the trivial double cover of the base.
Now focusing on the lower part of this figure:
\begin{center} 
\raisebox{-0.4\height}{\includegraphics[scale=0.25]{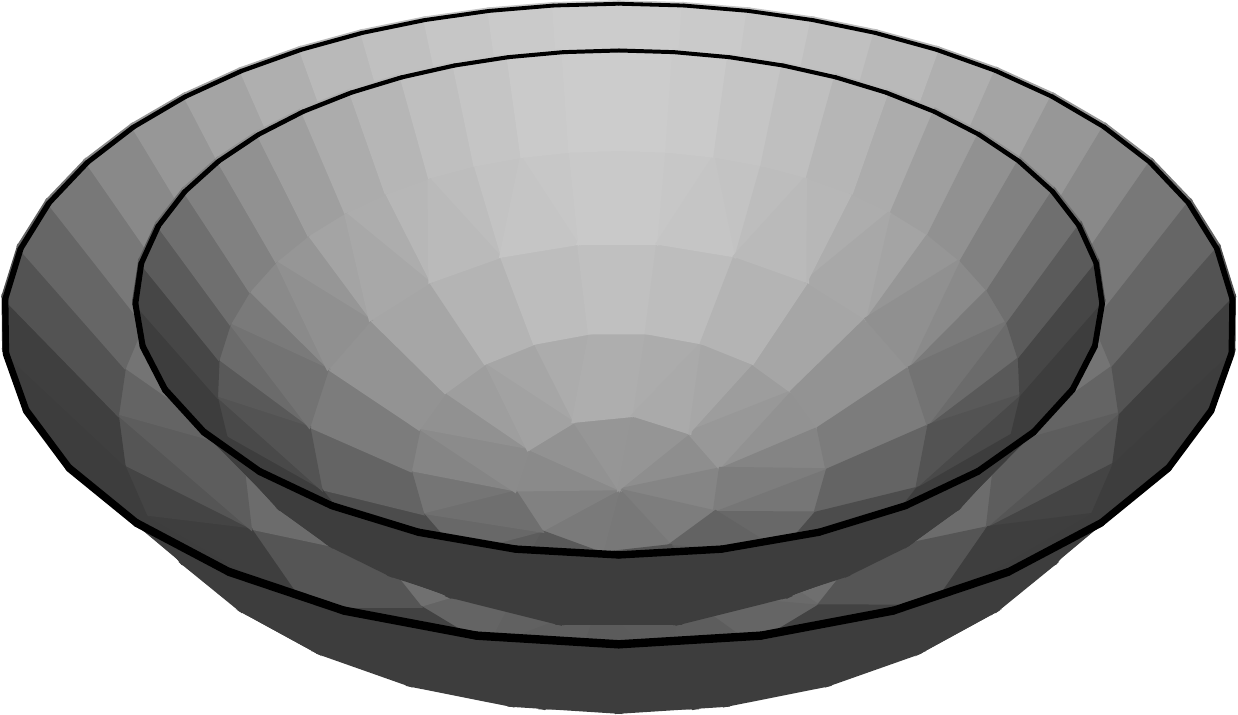}}
\ \ $\mapsto$\ \ 
\raisebox{-0.4\height}{\includegraphics[scale=0.25]{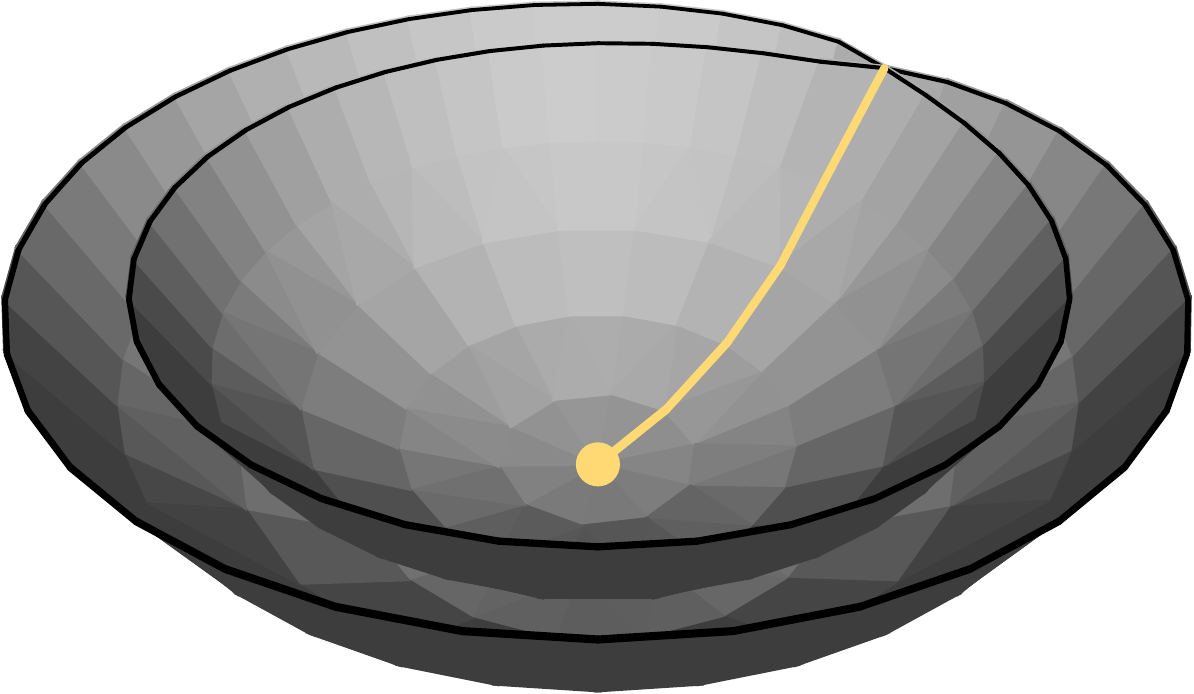}}
\end{center}
This shows explicitly how domain walls and
genons correspond to branches of double covers,
and will motivate the development of the qubit
theory below, see \Fref{fig:cover-genon}.

Here we tabulate values of $m$ for various double
covers:
$$
\begin{array}{c|cccccccc}
  m=\ ?    & g(E)\!=\!0 & g(E)\!=\!1 & g(E)\!=\!2 & g(E)\!=\!3 & g(E)\!=\!4 & g(E)\!=\!5 & & \\ 
\hline
g(B) = 0   &  2 &  4 &  6 &  8 &  10 &  12 \\
g(B) = 1   &    &  0 &  2 &  4 &  6  &  8  \\
g(B) = 2   &    &    &    &  0 &  2  &  4  \\
g(B) = 3   &    &    &    &    &     &  0  
\end{array}
$$
The interplay between the $m$ branch points and
the resulting genus of a double cover 
is the origin of the term \emph{genon}.
In summary, we have the following dictionary:
\begin{center}
\begin{tabular}{cc}
\underline{$D(\Integer_2)$ topological order} &
    \underline{Riemann surfaces} \\
domain wall  & branch cut \\
genon        & branch point 
\end{tabular}
\end{center}

\subsection{$ZX$-calculus survival guide}\label{sec:zx}

The $ZX$-calculus
is a notation for drawing quantum circuit diagrams \cite{Coecke2011,vandeWetering2020}.
We draw circuits from right to left, in agreement with algebraic (Dirac) notation.
The building building blocks for this
notation are wires, blue and green \emph{spiders}, 
and the yellow Hadamard box.
Such diagrams are examples of tensor networks.
In this section we give a brief and incomplete introduction to this notation
and how we compute with it.

\begin{center}
\begin{tabular}{|c|c|c|c|c|c|c|c|c|}
\hline
$ZX$-diagram 
& \igd{-0.25}{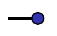}
& \igc{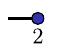}
& \igd{-0.25}{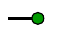}
& \igc{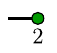}
& \igd{-0.25}{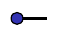}
& \igc{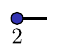}
& \igd{-0.25}{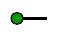}
& \igc{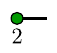}
\\
\hline
Dirac notation 
& $\ket{0}$
& $\ket{1}$
& $\ket{+}$
& $\ket{-}$
& $\bra{0}$
& $\bra{1}$
& $\bra{+}$
& $\bra{-}$
\\
\hline
\end{tabular}
\end{center}

\begin{center}
\begin{tabular}{|c|c|c|c|c|c|c|c|}
\hline
$ZX$-diagram 
& \igc{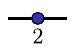}
& \igc{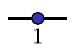}
& \igc{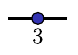}
& \igd{-0.2}{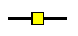}
& \igc{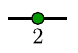}
& \igc{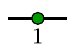}
& \igc{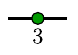}
\\
\hline
Gate
& $X$
& $\sqrt{X}$ 
& $\sqrt{X}^{\dag}$ 
& $H$ 
& $Z$ 
& $S$
& $S^{\dag}$
\\
\hline
\end{tabular}
\end{center}

\begin{center}
\begin{tabular}{|c|c|c|c|c|c|}
\hline
$ZX$-diagram 
& \igd{-0.35}{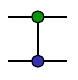}
& \igd{-0.35}{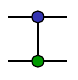}
& \igd{-0.35}{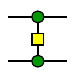}
& \igc{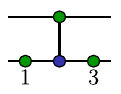}
& \igd{-0.35}{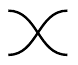}
\\
\hline
Gate
& $CX_{0,1}$ 
& $CX_{1,0}$ 
& $CZ$ 
& $CY_{0,1}$ 
& $SWAP$ 
\\
\hline
\end{tabular}
\end{center}

The definition of a blue or green \emph{spider} 
with $m$ outputs, $n$ inputs and labelled with \emph{phase} $a\in \Integer/4$ is
given by:
\begin{align*}
\igc{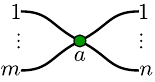} \ &:= \ 
\ket{0}^{\tensor m}\bra{0}^{\tensor n} + e^{2\pi ia/4} \ket{1}^{\tensor m}\bra{1}^{\tensor n} 
\\
\igc{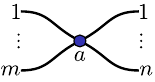} \ &:= \ 
\ket{+}^{\tensor m}\bra{+}^{\tensor n} + e^{2\pi ia/4} \ket{-}^{\tensor m}\bra{-}^{\tensor n} 
\end{align*}
When the phase is zero we usually omit the phase label.
ZX-diagrams without phase labels are called \emph{phase-free}.

The most important equations
for our purposes
allow us to commute the Pauli $X$ and $Z$ operators through a circuit.
When a phase $2$ operator meets a spider of the same colour,
it passes through freely:
\begin{equation*}
\begin{aligned}[c]
\igc{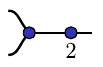} &= \igd{-0.6}{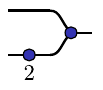} = \igc{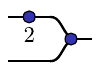} \\
\igc{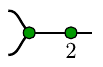} &= \igd{-0.6}{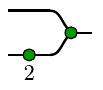} = \igc{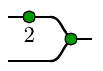} \\
\end{aligned}
\hsp
\begin{aligned}[c]
\igc{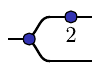} &= \igd{-0.6}{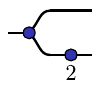} = \igc{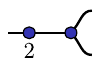} \\
\igc{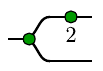} &= \igd{-0.6}{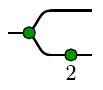} = \igc{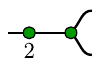} \\
\end{aligned}
\end{equation*}
When a phase $2$ operator meets a spider of the opposite colour,
it is copied:
\begin{equation*}
\begin{aligned}[c]
\igc{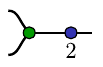} &= \igd{-0.55}{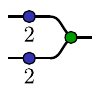} \\
\igc{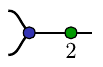} &= \igd{-0.55}{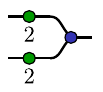} \\
\end{aligned}
\hsp
\begin{aligned}[c]
\igd{-0.55}{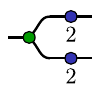} &= \igc{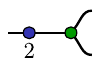} \\
\igd{-0.55}{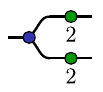} &= \igc{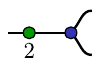} \\
\end{aligned}
\end{equation*}
Similarly,
states and effects get copied by spiders of the opposite colour:
\begin{equation*}
\begin{aligned}[c]
\igc{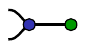} &= \igc{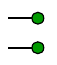} \\
\igc{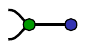} &= \igc{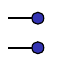} \\
\end{aligned}
\hsp
\begin{aligned}[c]
\igc{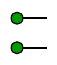} &= \igc{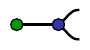} \\
\igc{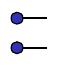} &= \igc{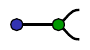} \\
\end{aligned}
\end{equation*}
Spider legs are flexible, and this is how we
justify the use of vertical wires in our $ZX$-diagrams. For example:
$$
 \igc{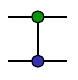} := \igc{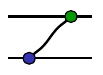} = \igc{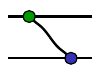}
$$
Using these rules we commute a Pauli $Z$ operator on the control
qubit of a $CX$:
$$
\igc{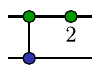} = \igc{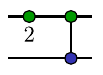}
$$
and a Pauli $X$ operator,
$$
\igc{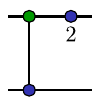} =
\igc{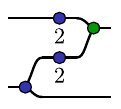} = 
\igd{-0.5}{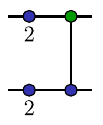}
$$
This diagrammatic calculus is blue-green symmetric, and indeed,
the $ZX$-calculus has a \emph{Fourier duality} which implements
this colour reversal:
$$
\igc{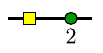} = \igc{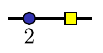}\hsp
\igc{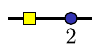} = \igc{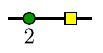}
$$
Now we can commute the Pauli $X$ operator through a $CZ$:
$$
\igc{zx_CZ_b2_1} =
\igc{zx_CZ_b2_2} = 
\igc{zx_CZ_b2_3} =
\igd{-0.5}{zx_CZ_b2_4}
$$

The Fourier duality plays a fundamental role in this work,
and we make use of this blue-green-yellow colour scheme consistently to
refer to this connection.

\subsection{Quantum codes and symplectic geometry}\label{sec:qcodes} 

In this section we recall basic facts about
qubit stabilizer codes, their relation to
symplectic geometry,
and the Clifford group~\cite{Calderbank1998,Haah2016}.

Given a field $\Field$, and integer $n\ge 0$,
we define the \emph{standard
$2n$-dimensional symplectic space} to be the $2n$-dimensional vector space
$\Field^{n}\oplus\Field^n$
with symplectic form 
$$
\Omega_n =
\begin{pmatrix}
0 & I_n \\
-I_n & 0
\end{pmatrix}
$$
where $I_n$ is the $n\times n$ identity matrix.
A vector $v\in \Field^{n}\oplus\Field^n$
is written is a block column matrix:
$$
v = \begin{pmatrix} v_X \\ v_Z \end{pmatrix}
$$
and we call $v_X$ the $X$ part of $v$ and $v_Z$ the $Z$ part of $v$.
Similarly for covectors and row matrices.

Given a vector $v$ in the vector space $\Field^n$ we define the \emph{weight} of $v$,
denoted $w(v)$,
to be the number of non-zero components of $v$.
For a vector $v$ in the standard symplectic space 
$\Field^{n}\oplus\Field^n$, we define the weight as
$$
w(v) = w(v_X) + w(v_Z) - w(v_X \cdot v_Z),
$$
where $v_X\cdot v_Z \in \Field^n$ is the componentwise product of $v_X$ and $v_Z$.

Let $\Field_2$ be the finite field of order $2$.
Much of the theory below can be developed for other finite fields,
but here we focus on the case $\Field_2$.
We define a \emph{quantum code} $C$ to be
an isotropic subspace $C\subset\Field_2^n\oplus\Field_2^n$
where $n\ge 0$ is an integer \emph{number of qubits}.
We also call $C$ the \emph{codespace}.
Given such a code, the \emph{logical space} is
the coisotropic subspace given by:
$$
C^{\perp} = \{v \in \Field_2^n\oplus\Field_2^n\ |\ v^{\top}\Omega_n C = 0\}
\supset C.
$$
The \emph{parameters} $[[n,k,d]]$ of a quantum code $C$ are:
$n$ the number of qubits, $k$ the dimension of $C^{\perp}/C$,
and $d$ the minimum weight of $v \in C^{\perp}$ with $v\notin C.$

Quantum codes $C$ can also be specified by 
a $m\times 2n$ \emph{(parity) check matrix} $\Parity$.
This is a full-rank matrix with 
$ \Parity\Omega_n \Parity^{\top} = 0. $
The codespace $C$ is the rowspan of $\Parity$,
and the logical space $C^\perp$ is the kernel (nullspace)
of the matrix $\Parity\Omega_n$.
We write such a check matrix in block form
$\Parity = \bigl( \Parity_X\ \Parity_Z \bigr)$ where $\Parity_X$ and $\Parity_Z$ are $m\times n$ matrices.
Expanding the isotropic condition 
we find the equivalent statement
\begin{align*}
\Parity_X \Parity_Z^\top - \Parity_Z\Parity_X^\top = 0.
\end{align*}
Given a quantum code $C\subset \Field_2^n\oplus\Field_2^n$
we say that $C$ is \emph{CSS} when we have
the direct sum decomposition $C=C_X\oplus C_Z$ with 
$C_X\subset\Field_2^n$
and
$C_Z\subset\Field_2^n$.
Equivalently, 
$C$ is CSS when it has a check matrix 
of the form
$$
\Parity = \bigl( \Parity_X\ \Parity_Z \bigr) = 
\begin{pmatrix}
\Parity'_X & 0 \\
0 & \Parity'_Z \\
\end{pmatrix}.
$$
In other words, $\Parity_X$ and $\Parity_Z$ can be written without any nonzero rows in common.

We will make use of \emph{Pauli operator} notation:
for a vector $v\in\Symp$, 
with components $(v_1,...,v_n, v_{n+1},...,v_{2n})$
we write this as a length $n$ string ($n$-tuple)
of symbols $I,X,Y,Z$ with $i$-th entry given by
$$
\left\{
\begin{array}{lll}
I & \text{if} & v_i=0, v_{n+i}=0, \\
X & \text{if} & v_i=1, v_{n+i}=0, \\
Z & \text{if} & v_i=0, v_{n+i}=1, \\
Y & \text{if} & v_i=1, v_{n+i}=1.
\end{array}
\right.
$$
We also use the dot $.$ in place of $I$.
For example, the vector $(1 0 1 1)\in \Field_2^2\oplus \Field_2^2$ 
has Pauli operator $YZ$.
The subspace of $\Symp$
spanned by $v_i, v_{n+i}$ is the \emph{$i$-th qubit}.
We declare two codes $C$ and $C'$ to be isomorphic $C\cong C'$
when they are the same up to permutation of qubits.

\begingroup
\renewcommand{\arraystretch}{0.6}
\setlength\arraycolsep{0pt}

For an example of a $[[4,1,2]]$ quantum code $C\subset \Field_2^4\oplus \Field_2^4$
we have the 
parity check matrix and corresponding Pauli operator notation:
$$
\Parity = 
\left(
\begin{array}{cccc;{1pt/0pt}cccc}
1&1&.&.&.&1&1&.\\.&1&1&.&.&.&1&1\\.&.&1&1&1&.&.&1
\end{array}
\right)
=
\begin{pmatrix}
{ X}&{ Y}&{ Z}&.\\
.&{ X}&{ Y}&{ Z}\\
{ Z}&.&{ X}&{ Y}\\
\end{pmatrix}.
$$
We have a vertical line separating the $\Parity_X$
and $\Parity_Z$ blocks,
and the dot notation is for zero or $I$ entries.
An example of a CSS code
$C\subset \Field_2^{8}\oplus \Field_2^{8}$ 
 with parameters $[[8,2,2]]$
is given by
$$
\Parity = 
\left(
\begin{array}{cccccccc;{1pt/0pt}cccccccc}
1&1&.&.&.&1&1&.&.&.&.&.&.&.&.&.\\
.&1&1&.&.&.&1&1&.&.&.&.&.&.&.&.\\
.&.&1&1&1&.&.&1&.&.&.&.&.&.&.&.\\
.&.&.&.&.&.&.&.&.&1&1&.&1&1&.&.\\
.&.&.&.&.&.&.&.&.&.&1&1&.&1&1&.\\
.&.&.&.&.&.&.&.&1&.&.&1&.&.&1&1
\end{array}
\right)
=
\begin{pmatrix}
{ X}&{ X}&.&.&.&{ X}&{ X}&.\\
.&{ X}&{ X}&.&.&.&{ X}&{ X}\\
.&.&{ X}&{ X}&{ X}&.&.&{ X}\\
.&{ Z}&{ Z}&.&{ Z}&{ Z}&.&.\\
.&.&{ Z}&{ Z}&.&{ Z}&{ Z}&.\\
{ Z}&.&.&{ Z}&.&.&{ Z}&{ Z}\\
\end{pmatrix}.
$$
These two examples are chosen for a reason:
the parity check matrix of the $[[8,2,2]]$
code contains two copies of the parity check matrix of
the $[[4,1,2]]$ code. 
This \emph{symplectic double} procedure is the subject of \S\ref{sec:sp} below.
\endgroup

\subsubsection{The qubit Clifford and Pauli groups}

\begingroup
\renewcommand{\arraystretch}{0.6}
\setlength\arraycolsep{4.8pt}

The $n$-qubit \emph{Pauli group},
also called the \emph{Heisenberg-Weyl group},
$\Pauli_2(n)$
is a subgroup of the unitary group $\Unitary(2^n)$ 
generated by $n$-fold tensor products of
the matrices
$$
iI = \begin{pmatrix}i & 0\\0 & i\end{pmatrix},\ \
X = \begin{pmatrix}0 & 1\\1 & 0\end{pmatrix},\ \
Z = \begin{pmatrix}1 & 0\\0 & -1\end{pmatrix}.
$$
This group has order given by
$$
|\Pauli_2(n)| = 4\cdot 4^{n}
$$
and the center $\Center(\Pauli_2(n))\cong \ZZ/4$
is generated by $i$.
The quotient $\Pauli_2(n)/\Center(\Pauli_2(n))$ is isomorphic to (the additive group of)
the $\Field_2$-vector space $\Field_2^{2n}.$
We write this as the short exact sequence:
$$
\ZZ/4 \rightarrowtail \Pauli_2(n) \twoheadrightarrow \Field_2^{2n}.
$$
The 2-cocycle for this central extension is a function
$\beta:\Field_2^{2n}\times\Field_2^{2n}\to \ZZ/4$
satisfying
$$\beta(v,w) \mod 2 = \langle v, w \rangle,$$
for all $v,w\in\Field_2^{2n}$.
Here we write $\langle v, w\rangle$
for the symplectic inner product on $\Field_2^{2n}$.
See \cite{Heinrich2021} \S 3.3.1.

The $n$-qubit Clifford group
can be defined to be the normalizer of $\Pauli_2(n)$ in the unitary group $\Unitary(2^n)$.
This is an infinite group, however for our purposes we
will be using the following finite subgroup
as our definition of 
the $n$-qubit \emph{Clifford group}.
This is generated from scalar and matrices,
$$
\omega,\ \ 
H = \frac{1}{\sqrt{2}}\begin{pmatrix}1 & 1\\1 & -1\end{pmatrix},\ \ 
S = \begin{pmatrix}1 & 0\\0 & i\end{pmatrix},\ \ 
CZ = \begin{pmatrix}1&0&0&0\\0&1&0&0\\0&0&1&0\\0&0&0&-1\end{pmatrix}
$$
using multiplication and tensor products.
For $n\ge 0$, this group is denoted $\Cliff_2(n)$
and has order given by
$$
|\Cliff_2(n)| = 8\prod_{j=1}^n 2(4^j - 1)4^j.
$$
This sequence begins as $8, 192, 92160, 743178240,...$ and is sequence A003956 in the OEIS.
These matrices have elements in the ring $\QQ[1^{1/8}].$
See \cite{Selinger2015}, Figure 8, for an abstract presentation of 
the Clifford group(s)
in terms of generators and relations.

The reference~\cite{Heinrich2021} uses a slightly different
definition of the Clifford group which is an index two subgroup of $\Cliff_2(n)$, \S 4.1.2.
\footnote{there's a typo in eq. (4.12)}
This is done by altering the definition of the Hadamard.
The generators are:
$$
i=\omega^2,\ \  \omega H=\frac{i+1}{2}\begin{pmatrix}1&1\\1&-1\end{pmatrix},
\ \  S,\ \  CZ.
$$
The generated matrices have elements in the ring $\QQ[i].$
We call this group the $n$-qubit \emph{semi-Clifford group}, denoted $\SCliff_2(n)$.
The OEIS notes that the order of these groups A027638
is also the order of a unitary group acting on Siegel modular forms.

The center $\Center(\Cliff_2(n))$ is isomorphic to $\ZZ/8$,
and we define the quotient group to be the \emph{affine symplectic group} over $\Field_2$:
$$
\ASp(2n,\Field_2) := \Cliff_2(n) / \mathcal{Z}(\Cliff_2(n)).
$$
Warning:
for $n>1$, this group is not (!) isomorphic to the
expected definition of the affine symplectic group which is the
semidirect product $\Sp(2n,\Field_2)\ltimes\Field_2^{2n}.$
This is a peculiar consequence of the dimension of 
our qubit space which is even. The story is much
simpler for odd prime-dimensional qudits.

Combining the above,
we have the following commutative diagram of group homomorphisms, where
every row and column is short exact:
$$
\begin{tikzcd}
\Field_2^{2n} \arrow[r, tail] & \ASp(2n, \Field_2) \arrow[r, two heads] & \Sp(2n, \Field_2) \\
\Pauli_2(n) \arrow[r, tail] \arrow[u, two heads] & \Cliff_2(n) \arrow[r, two heads] \arrow[u, two heads] & \Mp(2n, \Field_2) \arrow[u, two heads] \\
\ZZ/4 \arrow[r, tail]\arrow[u, tail] & \ZZ/8 \arrow[r, two heads]\arrow[u, tail] & \ZZ/2\arrow[u, tail]
\end{tikzcd}
$$
Here $\Mp(2n,\Field_2)$ is the metaplectic group over $\Field_2$.
This suggests that the Clifford group is, or should be called,
the affine metaplectic group.
See also \cite{Gurevich2012}.

In summary, the action of the Clifford group 
on the Pauli group by conjugation is,
up to phases, represented by the symplectic group.

\begin{figure*}[t]
\centering
\igr{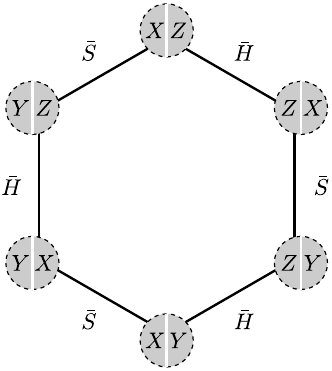}
\caption{
Here we record the action of single qubit Clifford operators
$S$ and $H$ on anti-commuting pairs of Pauli operators.
We notate the action via the barred operators $\bS$ and $\bH$
and we see that $\bS\bH\bS = \bH\bS\bH.$
}
\label{fig:lc}
\end{figure*}

If we omit the entangling gate $CZ$ from the list of
generators of the Clifford group, we get the
\emph{local Clifford group}.
As symplectic matrices, this is generated by
$n$-fold direct sums of 
the $\Field_2$ matrices
$$
\bH=\begin{pmatrix}0&1\\1&0\end{pmatrix},\ \ 
\bS=\begin{pmatrix}1&0\\1&1\end{pmatrix}.
$$
On a single qubit this group is $\Sp(2,\Field_2)$
which is isomorphic to the permutation group on
three elements $S_3$. See \Fref{fig:lc}.

\begin{lemma}
The $n$-qubit local Clifford group 
preserves the parameters $[[n,k,d]]$ of a quantum code.
\end{lemma}
\begin{proof}
The Clifford group preserves the parameters $n$ and $k$
of any quantum code $C$.
The local Clifford group preserves the weight of
any vector $v\in\Field_2^n\oplus\Field_2^n$ and in
particular will preserve the weight of any codeword in $C$, 
thereby preserving the parameter $d$.
\end{proof}

\endgroup

\section{The symplectic double}\label{sec:sp} 

Given a vector space $V$ over a field $\Field$,
we construct a
symplectic space $\Double(V) := V\oplus V^{\star}$ with symplectic
form:
\begin{align*}
    \Omega: \Double(V)\tensor \Double(V) &\to \Field \\
        (v\oplus c, u\oplus d) &\mapsto d(v) - c(u).
\end{align*}
Moreover, the assignment 
$$
\Double ( V ) = V\oplus V^{\star}
$$
is \emph{functorial}, 
which means that given invertible $f:V\to V$,
we have that 
\begin{align}\label{eq:double-f}
    \Double ( f  ) := f \oplus (f^{-1})^{\star}
\end{align}
is a symplectic map on $\Double(V)$:
\begin{align*}
    & \Omega( (f\oplus(f^{-1})^{\star})(v\oplus c), (f\oplus(f^{-1})^{\star})(u\oplus d) ) \\
    =\ & \Omega( f(v)\oplus cf^{-1}, f(u)\oplus df^{-1} ) \\
    =\ & d(v) - c(u)\\
    =\ & \Omega( v\oplus c, u\oplus d ).
\end{align*}
and also that $\Double(.)$ preserves composition.
In other words, we have a group homomorphism:
\begin{align}\label{eq:functorial}
\Double(.):\GL(V,\Field)\to \Sp(\Double(V),\Field)
\end{align}
and this homomorphism is injective.

When $V$ itself is symplectic, with symplectic form $\Omega_0$,
we have an isomorphism 
\begin{align*}
    V &\xrightarrow{\cong} V^\star \\
    v &\mapsto \Omega_0(v, \_).
\end{align*}
which we use in the definition of $\Double(.)$, as the following lemma shows.
This lemma is key to 
the construction of fault-tolerant Clifford gates
on doubled codes.
\begin{lemma}\label{lem:functorial}({\it Functoriality})
Given symplectic space $(V,\Omega_0)$ 
with $n$-dimensional isotropic subspace $U\subset V$
then $\Double(U) := U\oplus \Omega_0(U,\_)$
is a $2n$-dimensional isotropic subspace of $\Double(V)$.
Moreover, given a symplectic map $f:V\to V$ that
preserves $U$ as a subspace $f(U)=U$, then
$\Double(f)$ is a symplectic map that
preserves the subspace $\Double(U)$.
\end{lemma}
\begin{proof}
Clearly $\Double(U)$ is a subspace of $\Double(V)$,
what remains to be shown is that $\Double(U)$ is isotropic.
With $u,v,u',v'\in U$ we have
generic elements of $\Double(U)$ given by 
        $u\oplus\Omega_0(v,\_)$ and 
        $u'\oplus\Omega_0(v',\_)$. 
The symplectic pairing
evaluated on these two elements is 
\begin{align*}
    & \Omega( 
        u\oplus\Omega_0(v,\_),
        u'\oplus\Omega_0(v',\_)) \\
    =\ & \Omega_0(v',u) - \Omega_0(v,u')\\
    =\ & 0-0 = 0
\end{align*}
and so $\Double(U)$ is isotropic.
Next, the action $\Double(f):
        u\oplus\Omega_0(v,\_)
\mapsto
        f(u)\oplus\Omega_0(f^{-1}(v),\_)
\in\Double(U)$
and so $\Double(f)$ preserves the subspace $\Double(U)$
when $f$ preserves the subspace $U$.
\end{proof}

Given a $m\times 2n$ check matrix 
$\Parity = \bigl( \Parity_X\ \Parity_Z \bigr)$ 
the \emph{doubled} check matrix $\Double(\Parity)$ is
a $2m\times 4n$ check matrix
\begin{align}\label{eq:dh}
\Double(\Parity) := 
\begin{pmatrix}
\Parity_X & \Parity_Z & 0 & 0 \\
0 & 0 & \Parity_Z & \Parity_X 
\end{pmatrix}.
\end{align}
By direct calculation we verify this is the check
matrix for a quantum code (isotropic subspace), as promised by 
the functoriality lemma:
\begin{align}\label{eq:dhsymp}
\Double(\Parity)\Omega_{2n} \Double(\Parity)^\top = 
\begin{pmatrix}
0 & \Parity\Omega_n \Parity^\top \\
\Parity\Omega_n \Parity^\top & 0
\end{pmatrix} = 0.
\end{align}

\begin{theorem}\label{th:double}
Given a quantum code $C$ with parameters $[[n,k,d]]$,
we have $\Double(C)$ a CSS quantum code with
parameters $[[2n,2k,\ge d]]$.
\end{theorem}

\begin{proof}
By the functoriality lemma \ref{lem:functorial}, 
we see that $\Double(C)$ is a $2n$ qubit quantum code. 
A check matrix for the codespace $\Double(C)=C\oplus \Omega_n(C,\_)$
is given by \Eref{eq:dh}, which has CSS form.
Next, we examine logical operators $v\in\Double(C)^\perp$.
Both $v_X$ and $\Omega_n v_Z$ are in $C^\perp$,
and $w(v)\ge w(v_X) + w(v_Z)$, which is lower bounded by $d$
because one or both of $v_X,\Omega_n v_Z$ are not in $C$.
\end{proof}

A closely related result is the following fault-tolerant property of
fiber transversal Clifford gates.
\begin{theorem}\label{th:fault-tolerant}
Given a quantum code $C$ with parameters $[[n,k,d]]$ 
a logical Clifford on $\Double(C)$ that acts on each fiber separately
is fault-tolerant up to distance $d$.
\end{theorem}
\begin{proof}
This is very similar to the proof of Theorem \ref{th:double}.
See Appendix \ref{sec:proof}.
\end{proof}

\begin{claim}\label{th:double-css}
A quantum code $C$ is a CSS code iff 
$\Double(C) = C\oplus H(C).$ 
\end{claim}

We now give two operational characterizations of
the CSS codes that are the double of some other code.
The first characterization relies on the following definition.
Given an $n$ qubit CSS code $C=C_X\oplus C_Z$, 
a \emph{$ZX$-duality} on $C$
is a permutation $\tau:n\to n$ such that
$\tau(C_X)=C_Z, \tau(C_Z)=C_X$ where the action of
$\tau$ on subspaces of $\Field_2^n$ is by permuting coordinates.%
\footnote{
This is slightly more general than the definition in \cite{Breuckmann2022}.
}
The second characterization is in terms of concatenation
with a CSS code with stabilizers $XXXX,ZZZZ$, logicals $XXII,ZIZI,XIXI,ZZII$.
We call this \emph{the $[[4,2,2]]$} code, even though there are other CSS codes
with these parameters.

\begin{figure*}[t]
\centering
\begin{subfigure}[t]{0.3\textwidth}
\igr{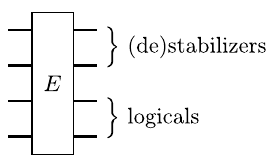}
\caption{The encoder $E$ has two qubit inputs for the
stabilizer/destabilizers, and two logical qubits.}
\label{fig:422-encode}
\end{subfigure}
~
\begin{subfigure}[t]{0.3\textwidth}
\!\igr{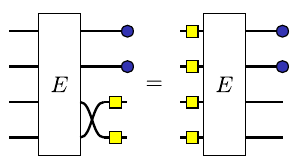}
\caption{$E$ exchanges a 2-qubit logical gate
with a physical transversal Hadamard}
\label{fig:422-ZX}
\end{subfigure}
~
\begin{subfigure}[t]{0.3\textwidth}
\igr{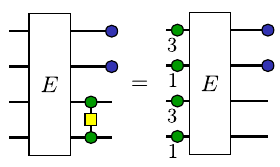}
\caption{$E$ exchanges a logical $CZ$ gate
with a physical transversal $S^{(\dag)}$ gate}
\label{fig:422-CZ}
\end{subfigure}
\caption{We show an encoding unitary $E$ for the $[[4,2,2]]$
code as a circuit diagram, which flows 
in the algebraic direction, from right to left.}
\label{fig:422}
\end{figure*}

\begin{theorem}\label{th:unwrap}
Given a CSS code $C$ on $2n$ qubits,
the following are equivalent:

(1) $C$ has a fixed-point-free involutory $ZX$-duality,

(2) $C=\Double(C')$ for some $n$ qubit
quantum code $C'$, and

(3) there is a concatenation
of $C$ with $n$ copies of the $[[4,2,2]]$
code that is self-dual.
\end{theorem}

\begin{proof}
\emph{(1)=>(2)}
Let $\tau:2n\to 2n$ be a fixed-point-free involutory $ZX$-duality on $C$.
This means that the orbits of $\tau$ have size two. 
Without loss of generality
we can assume these orbits are of the form $\{i,i+n\}_{i=1,...,n}$.
Let the check matrix for $C$ be given by
the $2m\times 4n$ matrix
$$
\Parity = 
\begin{pmatrix}
\Parity_X & 0 \\
0 & \Parity_Z \\
\end{pmatrix}
$$
We see that $\tau H_Z^\top = H_X^\top A$
where $A$ is an invertible $m\times m$ matrix.
Therefore, we have
$$
\begin{pmatrix}
A^\top\Parity_X & 0 \\
0 & \Parity_Z \\
\end{pmatrix}
$$
is in the form of a doubled check matrix \Eref{eq:dh}.

\emph{(2)=>(1)} The converse direction follows because 
a doubled check matrix always has the above $\tau$ a
(fixed-point-free involutory) $ZX$-duality.

\emph{(1)<=>(3)} 
Concatenation corresponds to composing encoding circuits.
The two qubit orbits of $\tau$ correspond to the pairs
along which we concatenate with the $[[4,2,2]]$ code.
The $[[4,2,2]]$ encoder satisfies the identity implementing a $ZX$-duality.
See \Fref{fig:422-ZX}.
\end{proof}

A stronger statement can be made: there is a bijection
between CSS codes $C$ with fixed-point-free involutory $ZX$-duality $\tau$,
and codes $C'$ such that $C\cong \Double(C')$.
In other words, there can be distinct codes $C'$ and $C''$ that
double to isomorphic codes $\Double(C')\cong\Double(C'')$.
We will see an example of this
in \S\ref{sec:xzzx} and \Fref{fig:ten-two-three}.

Given any of the conditions in Theorem \ref{th:unwrap}
together with a condition on the $X$-type stabilizers,
we have from \cite{Breuckmann2022} Theorem 7,
that $C$ will have a transversal $CZ$ gate.
Condition (3) is a generalization of the well-known construction
of the \{4,8,8\} colour code by concatenating 
the $[[4,2,2]]$ code with two copies of a toric
code paired along a $ZX$-duality \cite{Criger2016}.
We write this concatenation along a $ZX$-duality
$\tau$ as $[[4,2,2]]\otimes_{\tau}C$.

\begin{claim}
Given a quantum code $C$ on $n$ qubits
the following are equivalent: \\
(1) $\Double(C)$ has a fiber transversal CZ gate \\
(2) C has a basis $\{v_i\}$ such that the parity of $Y$'s in each $v_i$ is even \\
(3) $[[4,2,2]]\otimes_\tau\Double(C)$ has a transversal $S^{(\dag)}$ gate.
\end{claim}

\subsection{Lifting Cliffords}\label{sec:func} 

Recall that 
Cliffords in the phase-free $ZX$-calculus are generated by 
CX gates~\cite{Kissinger2022}.
The next theorem is a consequence of the functoriality of $\Double(\ )$.

\begin{theorem}\label{th:lift}
The injective group homomorphism 
$\Double:\Sp(2n,\Field_2) \to \Sp(4n,\Field_2)$
lifts to a homomorphism 
$\Double':\Sp(2n,\Field_2) \to \Cliff(2n)$
whose image is given by Clifford unitary gates in
the phase-free $ZX$-calculus with fixed-point-free involutory $ZX$-duality:
$$
\begin{tikzcd}
       & \Cliff(2n) \arrow[d, twoheadrightarrow] \\
\Sp(2n,\Field_2) \arrow[r, rightarrowtail, "\Double"]   
    \arrow[ur, rightarrowtail, "\Double'"]
    & \Sp(4n,\Field_2) 
\end{tikzcd}
$$
\end{theorem}
\begin{proof}
We define $\Double'$ on the generators (redundantly) as
\begin{align*}
\igc{green_111} & \mapsto \igc{gate_CX01} \\
\igc{blue_111} & \mapsto \igc{gate_CX10} \\
\igc{gate_H} & \mapsto \igc{gate_SWAP} \\
\igc{gate_CX01} & \mapsto \igc{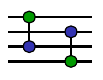} \\
\igc{gate_SWAP} & \mapsto \igc{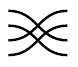} \\
\end{align*}
Note we are using the ``little-endian'' symplectic
convention for the string diagrams on the righ-hand-side.
This gives a (unitary) permutation representation of
$\Sp(2n,\Field_2)$ in the computational basis.
It is straightforward to check that this agrees with 
the application of \Eref{eq:double-f} to symplectic matrices $M$ on $\Symp$:
$$
\Double(M) = M \oplus (M^{-1})^{\top} .
$$
\end{proof}

For example,
given a code $C$ satisfying any
of the conditions of Theorem \ref{th:unwrap},
so that $C=\Double(C')$ 
a Hadamard on qubit $i$ in the base code $C'$
is lifted under $\Double'$ to 
swapping the qubits in $C$ in the fiber over $i$.
We will explore further examples in \S\ref{sec:examples} below.

This map $\Double'$ also appears in the proof of Theorem 3.8 in 
\cite{Backens2016}, there denoted as $[\![ \ ]\!]^{\natural}$.

\begin{claim}
The tanner graph of any symplectic
matrix $M\in\Sp(2n,\Field_2)$ gives a $ZX$-calculus diagram 
for $\Double'(M)$.
\end{claim}
See \Fref{fig:lc-spzx} for the single qubit symplectic matrices $\Sp(2,\Field_2)$
and the corresponding $ZX$-calculus diagrams.

\todo{show how to apply $\Double$ to encoders}

\todo{show how $\Double$ respects logical operators by Lemma \ref{lem:functorial} }

\todo{equation for $\Double(M)$ gives converse:
phase-free $ZX$-calculus diagrams with a symmetry property
correspond to diagrams that descend to the base code.}

\begin{figure*}[t]
\centering
\begin{subfigure}[b]{0.48\textwidth}
\igr{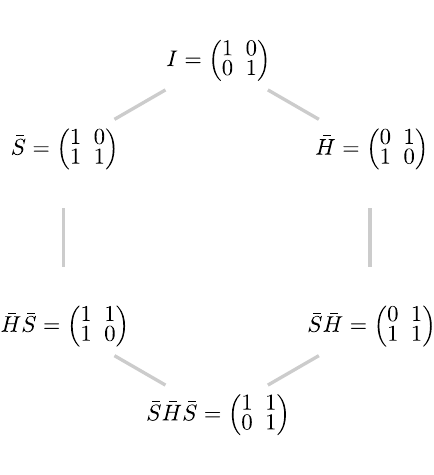} 
\caption{Symplectic matrices for the local Clifford group on one qubit}
\label{fig:lc-sp}
\end{subfigure}
~
\begin{subfigure}[b]{0.48\textwidth}
\igr{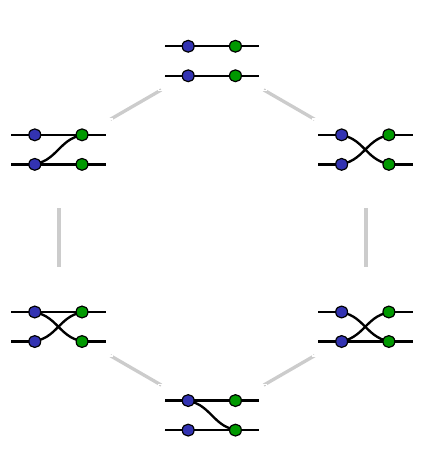}
\caption{The corresponding Tanner graph of the Symplectic matrices}
\label{fig:lc-zx}
\end{subfigure}
~
\caption{The Tanner graph for symplectic matrices
in $\Sp(2,\Field_2)$ 
gives the $ZX$-calculus
diagram for the lifted Clifford gate under $\Double'$.}
\label{fig:lc-spzx}
\end{figure*}

\section{Genon codes}\label{sec:genon} 

In this section we develop the theory of genon codes.
Examples are discussed below in \S\ref{sec:examples}.

\subsection{Genon graphs and genon codes}\label{sec:graphs} 

We are given a compact oriented surface $S$
with a finite graph $\Gamma$ embedded therein.
This subdivides the surface into 
\emph{faces}, 
\emph{edges}, and 
\emph{vertices}.
Vertices are required to have valence three or four.
Faces must have at least two edges, so that
bigon faces are allowed, but monogon faces are not.
We call such a $\Gamma$ a \emph{genon graph}.

A \emph{genon code} on a genon graph $\Gamma$ is a quantum code $C$ 
where qubits are placed at vertices,
stabilizers are placed on faces
with the following allowed configurations at each vertex:
\begin{center}
\raisebox{0.0\height}{\igr{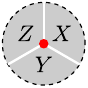}}\hsp
\igr{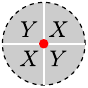}\hsp
\igr{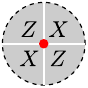}\hsp
\igr{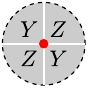}
\end{center}
We will write $(C,\Gamma)$ for a genon code $C$ on $\Gamma$.

\begin{theorem}
Given a genon graph $\Gamma$ with $n$ vertices,
there are $6n$ genon codes on $\Gamma$ and they are all 
equivalent under local Clifford operations.
\end{theorem}

\begin{proof}
The local Clifford group acts transitively
on the set of genon codes on $\Gamma$ because any vertex
configuration of valence $r=3,4$ is local Clifford equivalent
to any other vertex configuration of valence $r$.
Conversely, given a genon code,
the stabilizer subgroup of the local Clifford group 
is trivial and so we have the result that
there are $6n$ such distinct genon codes on a given graph $\Gamma$
with $n$ vertices.
\end{proof}

It's worthwhile staring at an illustration of this proof to
see how the local Clifford group acts on the 3-valent
and 4-valent vertices. 
You really do get 6 different configurations
at each vertex, and the local Clifford group moves between all of these:
\begin{center}
\igr{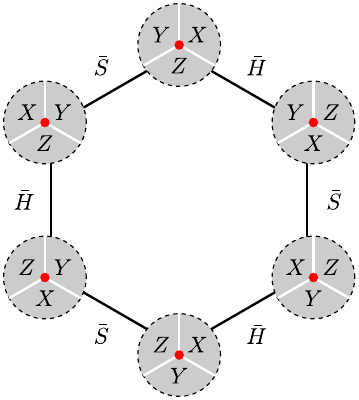}\Hsp
\igr{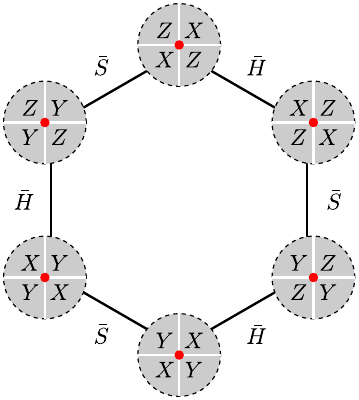}
\end{center}

\begin{lemma}
Let $C$ be a genon code on $\Gamma$ encoding $k$ logical qubits.
If $\Gamma$ is bicolourable
then $k=V-F+2$,
otherwise $k=V-F+1$.
\end{lemma}

\begin{proof}
For any stabilizer code with $n$ qubits and check matrix $\Parity$
we have that $k=n-\Rank(\Parity) = V -\Rank(\Parity).$
For a genon code $C$ on $\Gamma$, $\Parity$ is given by stabilizers living
on the faces.
Let $\Gamma$ be bicolourable, with faces either black or white.
Then $\Gamma$ has only four valent vertices, and we get
a linear dependent combination of white stabilizers, and also for
black stabilizers. Conversely, any linear dependent combination of
stabilizers is a sum of either of the black or white faces.
Therefore, $k=V-F+2$, and moreover, this argument runs backwards
so that $k=V-F+2$ implies bicolourable faces.
A similar argument shows that a lack of bicolourable faces is
equivalent to $k=V-F+1$. In this case the one linear dependency
comes from the combination of all the face stabilizers.
\end{proof}

\begin{theorem}
Let $C$ be a genon code on $\Gamma$ encoding $k$ logical qubits,
with $m$ the number of 3-valent vertices, and $g$ the genus of $\Gamma$.
If $\Gamma$ is bicolourable then $k=2g$ and $m=0$,
otherwise $k=2g + \frac{m}{2} - 1$.
\end{theorem}

\begin{proof}
Let $V_3$ be the number of 3-valent vertices, and $V_4$ be the number
of 4-valent vertices. 
Then we see the number of edges is 
$E = \frac{3}{2} V_3 + 2V_4$.
Writing the Euler characteristic:
\begin{align*}
\chi = 2-2g &= V - E + F \\
    &= V_3 + V_4 - \frac{3}{2}V_3 - 2V_4 + F \\
    &= F - \frac{1}{2}V_3 - V_4.
\end{align*}
If $\Gamma$ is bicolourable, then $V_3=0$ and
by the previous lemma we find $k=V_4-F+2$. Substituting
$F=2+V_4-k$ into the above equation for $\chi$ we get 
$2-2g = 2+V_4-k -V_4 $ and so $k=2g$.
If $\Gamma$ is not bicolourable the previous lemma
gives $F=V_3+V_4+1-k$ and so
$2-2g = V_3+V_4+1-k - \frac{1}{2}V_3 - V_4$ 
which gives $k=2g + \frac{m}{2} - 1$ as required.
\end{proof}

\subsection{String operators}\label{sec:string} 

See SY \cite{Sarkar2021} \S 4 and GS \cite{Gowda2020} \S III D.

We'd like to refer to logical operators of a genon code 
up to local Clifford equivalence.
This is a theory of \emph{string operators}
based only on the graph $\Gamma$. 
This will be a classical binary code $S$,
which is then linearly mapped onto the logical
space $C^\perp$ of a given genon code $C$ on $\Gamma$.

Given a genon graph $\Gamma$
we associate a vector space
$\Field_2^{2E}$ 
whose
basis is the edge-face pairs of $\Gamma$.
In other words, every edge has a vector space $\Field_2^2$
associated. 
We notate the four
vectors in this space with black lines
\begin{align*}
(0,0): \raisebox{-0.4\height}{\igr{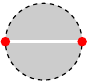}}\hsp
(1,0): \raisebox{-0.4\height}{\igr{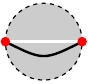}} \\
(0,1): \raisebox{-0.4\height}{\igr{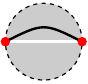}}\hsp
(1,1): \raisebox{-0.4\height}{\igr{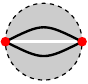}}
\end{align*}
and similarly for vectors in $\Field_2^{2E}$.
We now define the subspace $S\subset \Field_2^{2E}$
of \emph{string operators},
by considering the allowed string
operators around the 3-valent vertices and 4-valent vertices.
Around 3-valent vertices,
we have a vector space $\Field_2^6$, whose intersection with $S$ is
a 5-dimensional space spanned by 
\begin{center}
(i)
\raisebox{-0.4\height}{\igr{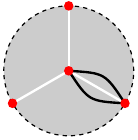}}\hsp
(ii)
\raisebox{-0.4\height}{\igr{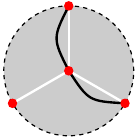}}\hsp
(iii)
\raisebox{-0.4\height}{\igr{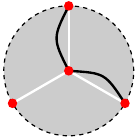}}
\end{center}
and rotations/reflections. 
These give all even parity vectors of $\Field_2^6$.
Around the 4-valent vertices we have a vector space 
$\Field_2^8$, whose intersection with $S$ is 
a 6-dimensional space spanned by
\begin{center}
(iv)
\raisebox{-0.4\height}{\igr{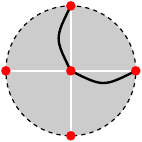}}\hsp
(v)
\raisebox{-0.4\height}{\igr{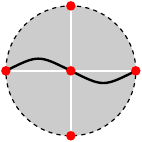}}\hsp
(vi)
\raisebox{-0.4\height}{\igr{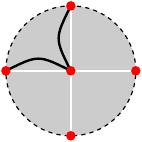}}
\end{center}
and rotations/reflections.
Note these diagrams are not all linearly independent,
for example (i)+(ii)=(iii) and (iv)+(v)=(vi).

Given a genon code $C$ on $\Gamma$
we define a linear map of string operators to logical operators
$$\phi:S\to C^\perp$$
on basis elements of $S$ as follows.
At a 4-valent vertex:
$$
\phi:\igc{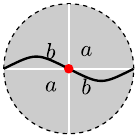} \mapsto a \Hsp
\phi:\igc{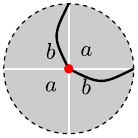} \mapsto a 
$$
and rotations/reflections.
At a 3-valent vertex:
$$
\phi:\igc{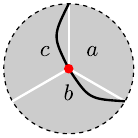} \mapsto a \Hsp
\phi:\igc{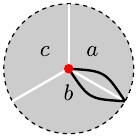} \mapsto c
$$
and rotations/reflections.
For example, linearity of $\phi$ implies the following:
$$
\phi:\igc{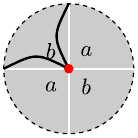} \mapsto 0 \Hsp
\phi:\igc{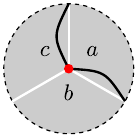} \mapsto b \Hsp
\phi:\igc{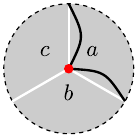} \mapsto 0 \Hsp
$$
Notice that the kernel of $\phi$ is non-trivial, 
in other words 
there is some redundancy in these string operators.

Using $\phi$ we can pick out a stabilizer generator $v\in C^\perp$
with a string operator \emph{external} to the corresponding face, 
however the string operator \emph{internal} to a face
is sent to zero, for example:
$$
\phi:\igc{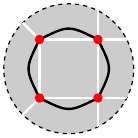}\mapsto v\in C^\perp \Hsp
\phi:\igc{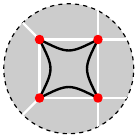}\mapsto 0\in C^\perp
$$
We summarize this in the following theorem.
\begin{theorem}
Given a genon code $C$ with parameters $[[n,k,d]]$, on a graph $\Gamma$
we have that 
$$
\dim(S) 
= \left\{\begin{array}{ll}
2k + 2F - 2 &\text{if $\Gamma$ is bicolourable}\\
2k + 2F -1  &\text{otherwise}
\end{array}\right.
= \left\{\begin{array}{ll}
2n + 2 &\text{if $\Gamma$ is bicolourable}\\
2n + 1  &\text{otherwise}
\end{array}\right.
$$
and
$\phi:S\to C^\perp$ is surjective with kernel spanned by
the internal face string operators.
\end{theorem}

\begin{proof}
Given a logical operator $v\in C^\perp$ we can construct
a string operator in $u\in S$ locally such that $\phi(u)=v$.
This is done by cases.
To find the kernel of $\phi$ we see that all the internal
face string operators are linearly independent, there are
$F$ many of these, where $F$ is the number of faces of $\Gamma$
and 
$$
F = \left\{\begin{array}{ll}
n-k+2 &\text{if $\Gamma$ is bicolourable}\\
n-k+1  &\text{otherwise}
\end{array}\right.
$$
\end{proof}
This theorem makes intuitive sense from the homological point of view:
stabilizer generators are given by operators that bound
a region, so they have an inside. 
Loosely speaking, the extra information found
in the string operators $S$ includes inside-out stabilizers,
which $\phi$ must send to zero.

Because the internal face string operators are sent to zero by $\phi$
we define the diagrammatic shorthand, or syntactic sugar:
$$
\igc{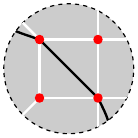} := 
\igc{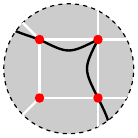} =_{\phi} \igc{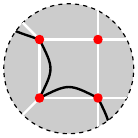}
$$
where the $\phi$ subscript refers to equality in the image of $\phi$.
In words, string operators can pass directly across faces 
and don't need to wind around the inside.
Examples of the use of this string diagram calculus are
given in \Fref{fig:xzzx}.

\subsection{Domain walls}\label{sec:domain} 

Given a genon code $C$ on a graph $\Gamma$,
we define a unique decoration on $\Gamma$,
called the \emph{domain walls} as follows:

(1) Every edge lies between two faces and we
place domain walls between the center of the two faces 
according to the rules:

\begin{center}
\igr{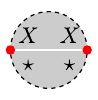}\hsp
\igr{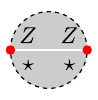}\hsp
\igr{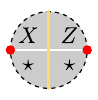}
\end{center}

\begin{center}
\igr{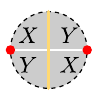}\hsp
\igr{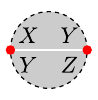}\hsp
\igr{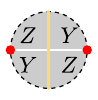}
\end{center}
where the star $\star$ denotes any Pauli consistent with the genon code rules.
For example a $YY$ configuration along one side of an edge is covered by these
rules because on the other side of the edge will be $XX$, $ZZ$ or $XZ$.

(2) Each face with a $Y$ at a trivalent vertex has a domain
wall from the center of the face to the $Y$ operator (the face-vertex flag):

\begin{center}
\igr{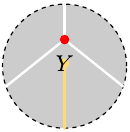}
\end{center}

\begin{theorem}
At the center of each face there is 
an even number of domain walls coming together.
\end{theorem}

\begin{proof}
Given a genon code $C$ we call the parity of a face
to be the number domain walls incident at the center mod 2.
Step 1:
we see that local Clifford operators preserve the domain wall parity
at each face. This is done by cases. 
Step 2: for each face, use local Clifford operators to 
convert this face stabilizer into a pure $X$ type stabilizer.
This has zero domain walls, which is even parity.
\end{proof}

From this theorem we see that if a domain wall has an endpoint
it will be at the $Y$ operator of a 3-valent vertex 
(and never at the center of a face).
We call these termination points \emph{genons}.

Given a genon code $C$ we see that there is one way to decorate
the surface $\Gamma$ with domain walls, however the converse is
not true.

\subsection{Double covers of genon codes}\label{sec:double} 

Given a genon code $C$ on a graph $\Gamma$, 
we define the \emph{double cover} of
$\Gamma$ relative to $C$ written $\Double(\Gamma,C)$, as follows:

(Dimension 0) 
Every vertex $v\in\Gamma$ is covered by two vertices in $\Double(\Gamma,C)$,
called the \emph{fiber} over $v$.
This fiber is ordered, and we write the two vertices over $v$ as $v^1$ and $v^2$.
See \Fref{fig:cover-genon-1}.

(Dimension 1) 
Every edge $e\in\Gamma$, with vertices $v_1,v_2$,
is covered by two edges in $\Double(\Gamma,C)$,
called the \emph{fiber} over $e$, written $e^1$ and $e^2$.
If $e$ does not cross a domain wall 
then $e^1$ has vertices $v_1^1,v_2^1$, and
$e^2$ has vertices $v_1^2,v_2^2$.
If $e$ does cross a domain wall
then $e^1$ has vertices $v_1^1,v_2^2$, and
$e^2$ has vertices $v_1^2,v_2^1$.
See \Fref{fig:cover-genon-2}.

(Genons) 
Every 3-valent vertex $v\in\Gamma$, with
incident face $f\in\Gamma$ whose stabilizer in $C$ supports a $Y$ operator
is covered by a single edge in $\Double(\Gamma,C)$
with vertices $v^1,v^2$.

(Dimension 2) 
Each face in $\Double(\Gamma,C)$ is constructed by lifting closed paths
$\gamma$ around the edges of a face $f\in\Gamma$.
The lifted edges
in these lifted paths then form the edges of a face.
When the path $\gamma$ encounters a genon at $v$, the edge between
$v^1,v^2$ is included in the lifted face, see \Fref{fig:cover-genon-3}
If the parity of the domain walls around $f$ is even there will
be two lifted faces $f^1,f^2$, 
coming from two possible lifts of $\gamma$,
otherwise there is only one lifted face $f^1$ which comes from
the unique lift of $\gamma$.

\begin{figure*}[t]
\centering
\begin{subfigure}[t]{0.2\textwidth}
\Hsp\igr{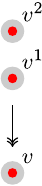} 
\caption{Every vertex is covered by a pair of vertices.}
\label{fig:cover-genon-1}
\end{subfigure}
~
\begin{subfigure}[t]{0.3\textwidth}
\hsp\igr{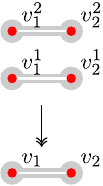} \hsp
\igr{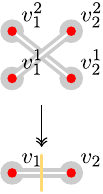} 
\caption{Every edge is covered by a pair of edges, 
with endpoint vertices swapped across domain walls.}
\label{fig:cover-genon-2}
\end{subfigure}
~
\begin{subfigure}[t]{0.45\textwidth}
\igr{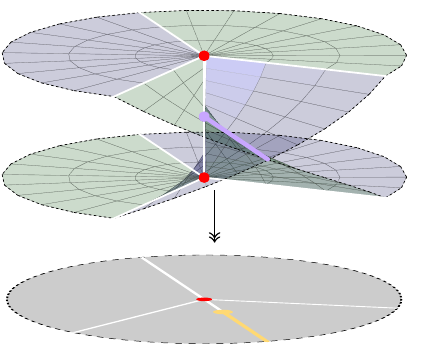}
\caption{Trivalent vertex and associated branch point 
at the other end of a ``half-edge''.
This half-edge becomes a proper edge in the double cover.
We show 
the lift of the domain wall in purple. 
The lifted faces are bicoloured green and blue.}
\label{fig:cover-genon-3}
\end{subfigure}
\caption{The double cover of a genon graph $\Gamma$
relative to a code $C$ is constructed
dimension wise, starting with vertices, then edges, then faces.}
\label{fig:cover-genon}
\end{figure*}

Given a genon code $C$ on $\Gamma$ we say that $C$ is
\emph{clean} when the stabilizers around 4-valent vertices support
only $X$-type and $Z$-type operators. 
In this sense, there are no unnecessary $Y$ operators.

\begin{claim}\label{th:clean}
Given a clean genon code $C$ on $\Gamma$,
then $\Double(C)$ is a CSS genon code on $\Double(\Gamma,C)$.
\end{claim}

\begin{claim}
Given a clean genon code $C$ on $\Gamma$,
then $\Double(\Gamma,C)$ is bicolourable and supports two CSS
genon codes, one of which is $\Double(C)$ and the other its dual.
\end{claim}

\begin{claim}
Given a genon code $C$ on $\Gamma$,
with parameters $[[n,k,d]]$ then 
$[[4,2,2]]\otimes_{\tau} \Double(C)$ is a self-dual $[[4n,2k,\ge d]]$ code.
\end{claim}

\begin{claim}
Given a genon code $C$ on $\Gamma$,
then $\Gamma$ is bicolourable iff 
$[[4,2,2]]\otimes_{\tau} \Double(C)$ is a colour code.
\end{claim}

\section{Example genon codes and protocols}\label{sec:examples}  

\subsection{Genus zero}

For this example, 
the genon graph $\Gamma$ is a tetrahedron inscribed on a sphere:
there are four qubits and four face stabilizers,
see \Fref{fig:cover-412}.
A nice choice for $C$ is given by 
the redundant stabilizer group generators:
$$
\langle XYZI, IXYZ, ZIXY, YZIX \rangle.
$$
Evidently, this code has a $\Integer/4$ cyclic symmetry 
$1\mapsto 2\mapsto 3 \mapsto 4\mapsto 1$,
whose action on the logical operators
$$
L_X = ZXII, \ L_Z = IZXI
$$
is a logical Hadamard gate:
\begin{align*}
L_X & \mapsto IZXI = L_Z \\
L_Z & \mapsto IIZX = L_X\cdot(ZIXY)\cdot(IXYZ).
\end{align*}
Moreover, it turns out we can perform any
of the $4!=24$ permutations on the physical qubits
followed by local Clifford gates, and these
yield all the single qubit logical Clifford gates for this code.
We tabulate the complete protocol:
$$
\begin{array}{c|c|c}
\text{permutation} & \text{local Clifford gate} & \text{logical gate} \\
\hline
(1, 2, 3, 4) & IIII & I \\
(1, 2, 4, 3) & \bS(\bH\cdot \bS\cdot \bH)(\bS\cdot \bH)(\bH\cdot \bS) & (\bH\cdot \bS\cdot \bH) \\
(1, 3, 2, 4) & (\bH\cdot \bS\cdot \bH)(\bS\cdot \bH)(\bH\cdot \bS)\bS & \bS \\
(1, 3, 4, 2) & (\bH\cdot \bS)\bS(\bH\cdot \bS\cdot \bH)(\bS\cdot \bH) & (\bH\cdot \bS) \\
(1, 4, 2, 3) & (\bS\cdot \bH)(\bH\cdot \bS)\bS(\bH\cdot \bS\cdot \bH) & (\bS\cdot \bH) \\
(1, 4, 3, 2) & \bH\bH\bH\bH & \bH \\
(2, 1, 3, 4) & (\bS\cdot \bH)(\bH\cdot \bS)\bS(\bH\cdot \bS\cdot \bH) & (\bH\cdot \bS\cdot \bH) \\
(2, 1, 4, 3) & \bH\bH\bH\bH & I \\
(2, 3, 1, 4) & \bS(\bH\cdot \bS\cdot \bH)(\bS\cdot \bH)(\bH\cdot \bS) & (\bS\cdot \bH) \\
(2, 3, 4, 1) & IIII & \bH \\
(2, 4, 1, 3) & (\bH\cdot \bS)\bS(\bH\cdot \bS\cdot \bH)(\bS\cdot \bH) & \bS \\
(2, 4, 3, 1) & (\bH\cdot \bS\cdot \bH)(\bS\cdot \bH)(\bH\cdot \bS)\bS & (\bH\cdot \bS) \\
(3, 1, 2, 4) & (\bH\cdot \bS)\bS(\bH\cdot \bS\cdot \bH)(\bS\cdot \bH) & (\bH\cdot \bS) \\
(3, 1, 4, 2) & (\bH\cdot \bS\cdot \bH)(\bS\cdot \bH)(\bH\cdot \bS)\bS & \bS \\
(3, 2, 1, 4) & \bH\bH\bH\bH & \bH \\
(3, 2, 4, 1) & (\bS\cdot \bH)(\bH\cdot \bS)\bS(\bH\cdot \bS\cdot \bH) & (\bS\cdot \bH) \\
(3, 4, 1, 2) & IIII & I \\
(3, 4, 2, 1) & \bS(\bH\cdot \bS\cdot \bH)(\bS\cdot \bH)(\bH\cdot \bS) & (\bH\cdot \bS\cdot \bH) \\
(4, 1, 2, 3) & IIII & \bH \\
(4, 1, 3, 2) & \bS(\bH\cdot \bS\cdot \bH)(\bS\cdot \bH)(\bH\cdot \bS) & (\bS\cdot \bH) \\
(4, 2, 1, 3) & (\bH\cdot \bS\cdot \bH)(\bS\cdot \bH)(\bH\cdot \bS)\bS & (\bH\cdot \bS) \\
(4, 2, 3, 1) & (\bH\cdot \bS)\bS(\bH\cdot \bS\cdot \bH)(\bS\cdot \bH) & \bS \\
(4, 3, 1, 2) & (\bS\cdot \bH)(\bH\cdot \bS)\bS(\bH\cdot \bS\cdot \bH) & (\bH\cdot \bS\cdot \bH) \\
(4, 3, 2, 1) & \bH\bH\bH\bH & I \\
\end{array}
$$ 
Using this protocol we can lift all of these
gates to logical Clifford gates on the $[[8,2,2]]$ code
by Theorem \ref{th:lift}.
We see the $(1,4,3,2)$ permutation for logical $\bH$ and the
$(1,3,2,4)$ for logical $\bS$,
as well as $(2,1,4,3)$ and $(3,4,1,2)$ for logical $I$,
agree with the anyon calculations in \S\ref{sec:genus-zero}.
The two logical gates $\bH,\bS$ generate the whole single qubit Clifford group
and will be used in the experiments \S \ref{sec:experiments} below.
We also note this $[[4,1,2]]$ code is local Clifford equivalent to the 
$[[4,1,2]]$ triangular colour code presented in \cite{Kesselring2018}, \S 5.2.

The surface codes, such as the $[[5,1,2]]$ code 
are genon codes on a sphere, \Fref{fig:genus-zero} (i). 
The missing external stabilizer forms 
the back of the sphere, and contains the domain walls.
In general, any surface code with a single boundary component
forms a genon code in this way.

As a next example we take Landahl's jaunty code \cite{Landahl2020}.
This is a $[[14,3,3]]$ code on a rhombic dodecahedron inscribed on a sphere,
\Fref{fig:genus-zero}(iv).

Below we tabulate some of these examples and their symplectic doubles.
\begin{center}
\begin{tabular}{ccl|cc}
base code   & genons & see Fig.  & symplectic double & genus \\
\hline
$[[4,1,2]]$  &  $m=4$  &  \ref{fig:cover-412} & $[[8,2,2]]$  & $g=1$ \\
$[[5,1,2]]$  &  $m=4$  &  \ref{fig:genus-zero} (i)\&(ii) & $[[10,2,3]]$  & $g=1$ \\
$[[6,2,2]]$  &  $m=6$  &  \ref{fig:genus-zero} (iii) & $[[12,4,2]]$  & $g=2$ \\
$[[14,3,3]]$ &  $m=8$  &  \ref{fig:genus-zero} (iv)\&(v) & $[[28,6,3]]$  & $g=3$ \\
\end{tabular}
\end{center}

\begin{figure*}[t]
\centering
\begin{subfigure}[t]{0.2\textwidth}
\includegraphics[scale=0.6]{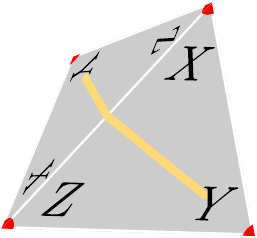}
\caption{We take a genon graph to be a tetrahedron inscribed on a sphere.}
\end{subfigure}
~
\begin{subfigure}[t]{0.3\textwidth}
\igr{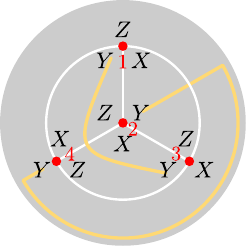}
\caption{The same graph with one face splayed out to infinity. Qubits
are numbered in red 1-4.}
\end{subfigure}
~
\begin{subfigure}[t]{0.4\textwidth}
\hsp\igr{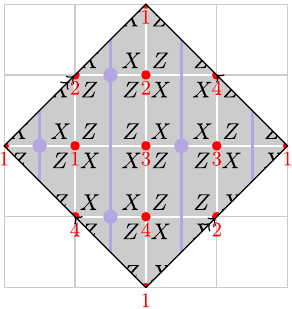}
\caption{The double cover, with qubits numbered as to which qubit
is covered in the base $[[4,1,2]]$ code.
The covering branch points (4 of them) and domain walls are shown in purple.
}
\end{subfigure}
\caption{A genus zero $[[4,1,2]]$ genon code is double covered by
an $[[8,2,2]]$ toric code. 
}
\label{fig:cover-412}
\end{figure*}

\begin{figure*}[]
\centering
\begin{subfigure}[t]{0.4\textwidth}
\igr{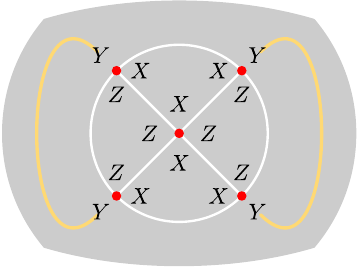} 
\caption{This is the well-known $[[5,1,2]]$ surface code.
}
\end{subfigure}
~
\begin{subfigure}[t]{0.4\textwidth}
\Hsp\igr{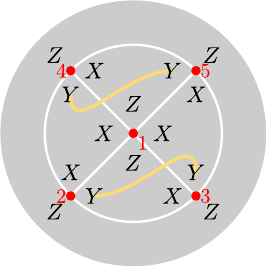} 
\caption{Another $[[5,1,2]]$ code on the same graph. 
This is a non-CSS code, local Clifford equivalent to (i).
}
\end{subfigure}
\begin{subfigure}[t]{0.8\textwidth}
\Hsp\Hsp\Hsp\igr{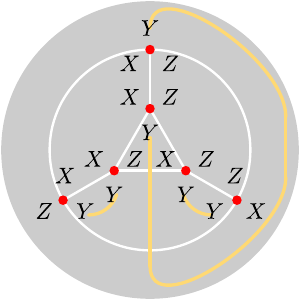} 
\caption{
This is a $[[6,2,2]]$ code.
Qubits are placed on the vertices of a triangular
prism. There are two triangular faces and three square faces.
}
\end{subfigure}
\begin{subfigure}[t]{0.4\textwidth}
\igr{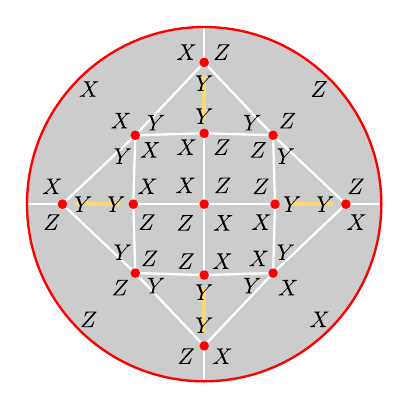} 
\caption{Landahl's jaunty $[[14,3,3]]$ code
We show this in splayed view with one qubit stretched
out to infinity.
}
\end{subfigure}
~
\begin{subfigure}[t]{0.4\textwidth}
\igr{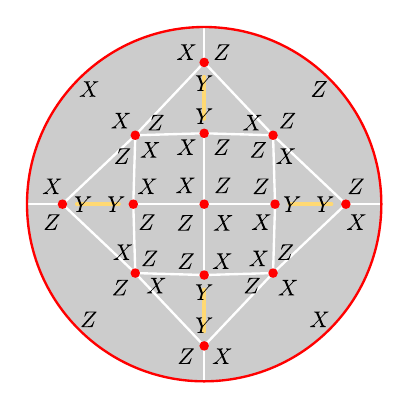} 
\caption{A local Clifford equivalent $[[14,3,3]]$
is in a clean configuration which means the
valence 4 vertices see only $X$ and $Z$ operators,
and ensures the doubled code is a genon code, Thm.\ref{th:clean}.
}
\end{subfigure}
\caption{Various genon codes that have genus zero.}
\label{fig:genus-zero}
\end{figure*}

\subsection{Genus one }\label{sec:xzzx} 

\begin{figure*}[]
\centering
\begin{subfigure}[t]{0.22\textwidth}
\igr{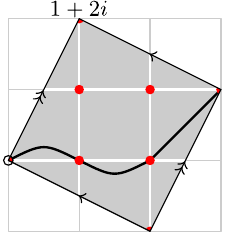} 
\caption{ }
\end{subfigure}
~
\begin{subfigure}[t]{0.3\textwidth}
\igr{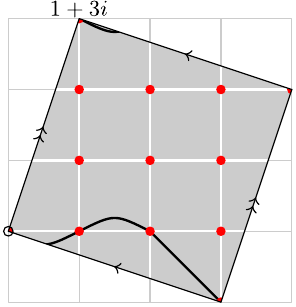}
\caption{ }
\end{subfigure}
~
\begin{subfigure}[t]{0.4\textwidth}
\igr{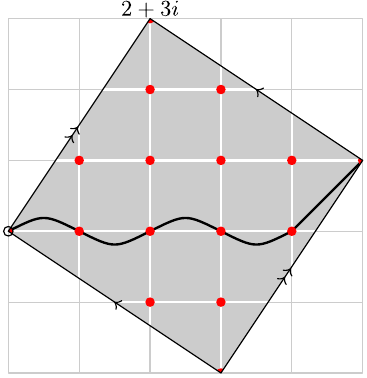}
\caption{ }
\end{subfigure}
\begin{subfigure}[t]{0.22\textwidth}
\igr{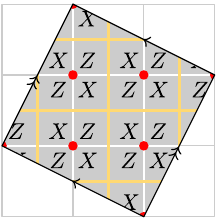}
\caption{A $[[5,1,3]]$ code.}
\end{subfigure}
~
\begin{subfigure}[t]{0.3\textwidth}
\igr{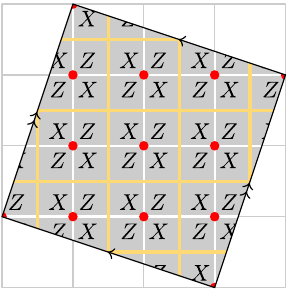}
\caption{A $[[10,2,3]]$ code.}
\end{subfigure}
~
\begin{subfigure}[t]{0.4\textwidth}
\igr{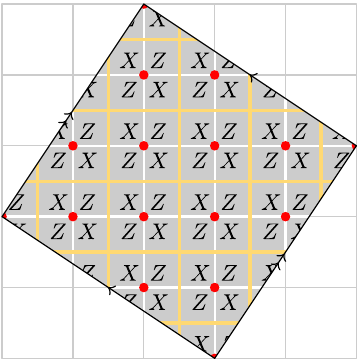}
\caption{A $[[13,1,5]]$ code.}
\end{subfigure}
\begin{subfigure}[t]{0.22\textwidth}
\igr{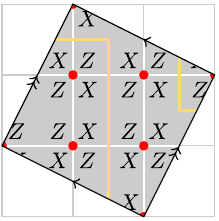}
\caption{ }
\end{subfigure}
~
\begin{subfigure}[t]{0.3\textwidth}
\igr{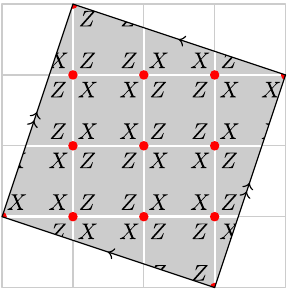}
\caption{ }
\end{subfigure}
~
\begin{subfigure}[t]{0.4\textwidth}
\igr{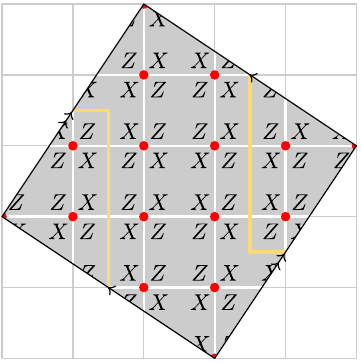}
\caption{ }
\end{subfigure}
\caption{
(i)-(iii): genon graphs built from quotients 
$\Integer[i]/\langle a+bi\rangle$ and logical string operators.
We marked the origin as well as the Gaussian integer $a+bi$.
(iv)-(vi): the $XZZX$ code and corresponding domain walls.
(vii)-(ix): local Clifford equivalence can remove (some) domain walls.
}\label{fig:xzzx}
\end{figure*}

We parameterize these by two integers $(a,b)$.
We view these periodic lattices 
as quotients of the Gaussian integers: $\Integer[i]/\langle a + bi\rangle.$
See \Fref{fig:xzzx}.
The resulting code
has parameters $[[n,k,d]]$ with 
$$
n = a^2 + b^2, \ \ 
k = \left\{\begin{array}{ll}1 & \text{if $n$ odd} \\ 2 & \text{if $n$ even}\end{array}\right.
,\ \ 
d = \left\{\begin{array}{ll}a+b & \text{if $n$ odd} \\ \max(a,b) & \text{if $n$ even}\end{array}\right.
.
$$
See \cite{Sarkar2021}, Theorem 3.9.,
where they define
a family of codes called the cyclic toric codes when $\gcd(a,b)=1$.

We make a table of some of these:
$$
{\small 
\begin{array}{r|rrrrrrrr}
  &          a=0  &         a=1   &          a=2   &      a=3      &      a=4      &        a=5    &        a=6     &       a=7    \\
\hline 
b=2 &[[4, 2, 2]]  & [[5, 1, 3]]  &[[8, 2, 2]]  &             &             &              &            &      \\
b=3 &[[9, 1, 3]]  & [[10, 2, 3]] &[[13, 1, 5]] &[[18, 2, 3]] &             &              &            &    \\
b=4 &[[16, 2, 4]] & [[17, 1, 5]] &[[20, 2, 4]] &[[25, 1, 7]] &[[32, 2, 4]] &              &            &    \\
b=5 &[[25, 1, 5]] & [[26, 2, 5]] &[[29, 1, 7]] &[[34, 2, 5]] &[[41, 1, 9]] &[[50, 2, 5]]  &            &     \\
b=6 &[[36, 2, 6]] & [[37, 1, 7]] &[[40, 2, 6]] &[[45, 1, 9]] &[[52, 2, 6]] &[[61, 1, 11]] &[[72, 2, 6]]&       \\
b=7 &[[49, 1, 7]] & [[50, 2, 7]] &[[53, 1, 9]] &[[58, 2, 7]] &[[65, 1, 11]]& [[74, 2, 7]] &[[85,1,13]]  &  [[98,2,7]]   \\
\end{array}
}
$$

\begin{figure*}[]
\centering
\begin{subfigure}[t]{0.4\textwidth}
\igr{base_512}
\caption{The genus zero base code with parameters $[[5,1,2]]$ }
\end{subfigure}
\begin{subfigure}[t]{0.4\textwidth}
\igr{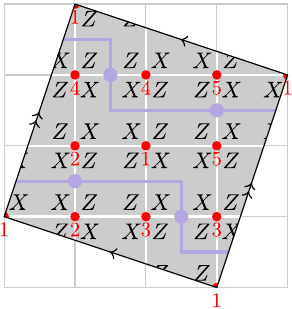}
\caption{The $[[10,2,3]]$ double cover.}
\end{subfigure}
\begin{subfigure}[t]{0.4\textwidth}
\igr{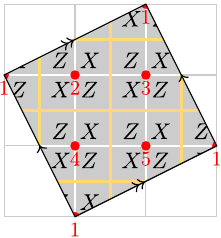}
\caption{The genus one base code is the $[[5,1,3]]$ code.}
\end{subfigure}
\begin{subfigure}[t]{0.4\textwidth}
\igr{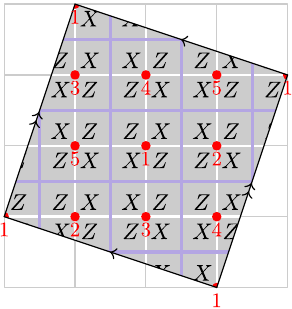}
\caption{The $[[10,2,3]]$ double cover.}
\end{subfigure}
\caption{
On the $[[10,2,3]]$ code 
there is a total of six $ZX$-dualities satisfying the requirements of 
Theorem \ref{th:unwrap}.
Five of these correspond to genus zero base code, and the other is the
genus one base code.
Here we show an example of a genus zero base code (i) as covered by (ii),
as well as a genus one code (iii) as covered by (iv).
The qubits are enumerated in each base, and these numbers are lifted into the
respective cover.
}
\label{fig:ten-two-three}
\end{figure*}

So far all these genus one codes have no genons.
In \Fref{fig:sixgenon} we show a $[[12,4,3]]$ genus one code with six genons.

\begin{figure*}[]
\centering
\igr{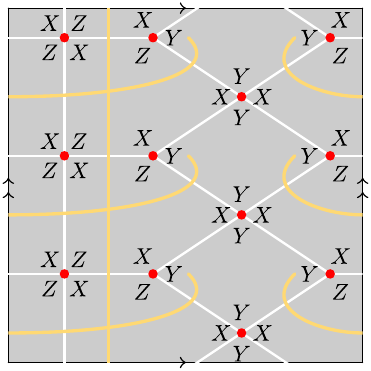}
\caption{This genus one code has six genons and parameters $[[12,4,3]]$.}
\label{fig:sixgenon}
\end{figure*}

\section{Experimental results}\label{sec:experiments} 
The permutations required for implementing genon protocols
and fault-tolerant gates resulting from
the lifting theorem, can be efficiently realized on a hardware with high-connectivity.
In architectures with fixed qubit structures and thus restricted connectivity, qubit permutations are realized through circuit-level swapping, moving information between two separate qubits by performing a series of local swaps of adjacent qubits between them, thus increasing the overall circuit depth.
In systems with arbitrary connectivity, qubit permutations can be realized through a simple "relabelling" of the physical qubits. 

Quantinuum's trapped-ion H1-1 device is based the quantum CCD (QCCD) architecture \cite{Wineland98, Pino2021} realizes all-to-all connectivity by using ion-transport operations. 
Ions are able to physically move around the linear trap, physically swapping locations as needed. 
As such, the H1-1 realizes fault-tolerant genon protocols with little overhead, incurring only noise from transport as no circuit-level swapping is required.

The H1-1 device uses 20 $^{171}$Yb$+$ ions for physical qubits and 20 $^{138}$Ba$+$ ions for sympathetic cooling, totally to 40 physical ions in the trap.  Gates are executed via stimulated Raman transitions implemented with off-resonant lasers directed in five distinct regions. 
During the course of these experiments, the average physical fidelities for the single-qubit and two-qubit gates, as measured by randomized benchmarking~\cite{Magesan2011} experiments averaged over all five gate zones, were $3.2(5)\times 10^{-5}$ and $9.2(5) \times 10^{-4}$ respectively. 
State preparation and measurement errors were also measured to be $2.7(1) \times 10^{-3}$.

Idling and transport error are more difficult to more difficult to characterize in the QCCD architecture as different circuits require different transport sequences. 
Indeed, depending on the circuit, certain permutations may not require additional transport. 
This gives the opportunity but not a prior guarantee for the compiler to reduce transport costs to next to zero, potentially allowing for very low error rate over heads. But more work needs to be done to study transport overheads for these sort of logical gates in practice. 
We leave it to future work to characterize how such transport impacts the logical gate fidelity of the protocol realized. For further description of the H1-1 hardware benchmarks and specifications, see ~\cite{Pino2021, qtuumspec}.

In this work, we realized three different experimental implementations of the theory work outlined above on the H1-1 QCCD device.
First, we realized genon braiding in the $[[4,1,2]]$, done by performing local Cliffords and qubit permutations. Such permutations are physically realized through ion-transport.
We demonstrate the ability to produce the full Clifford group through the demonstration of \textit{logical} randomized benchmarking, thus showcasing the power of the genon braiding technique. 
Next, by lifting the logical $S$ gate from the $[[4,1,2]]$ code, we benchmark a logical $CX$ gate of the [[8,2,2]] code, the double cover of the $[[4,1,2]]$. This logical gate, involving qubit permutations, again is efficiently realized through the qubit relabelling enabled by ion transport primitives. 
Finally, we realize another implementation two qubit logical gate from lifting the transversal $SH$ gate on the $[[5,1,3]]$ code. We benchmark this gate in a similar manner to the $CX$ on the $[[8,2,2]]$ but this time require proper decoding (rather than post-selection).

\subsection{The [[4,1,2]] protocol}
\begin{figure*}[]
\centering
\begin{subfigure}[t]{0.45\textwidth}
\igr{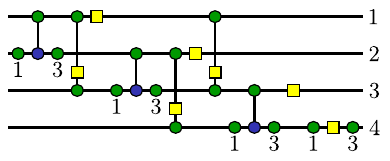}
\caption{Unitary encoding circuit. 
Input qubits are numbered $1-4$.
(De-)stabilizer inputs are qubits $1-3$, 
and the logical input qubit is qubit $4$.
}
\end{subfigure}
\begin{subfigure}[t]{0.35\textwidth}
\igr{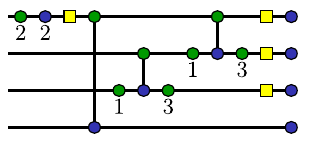}
\caption{Preparing a logical $\ket{0}$ state.}
\end{subfigure}
\begin{subfigure}[t]{0.25\textwidth}
\igr{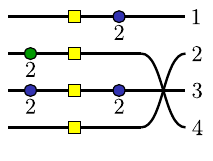}
\caption{Logical $H$ gate from swapping genons on qubits 2 and 4.}
\end{subfigure}
~
\begin{subfigure}[t]{0.25\textwidth}
\igr{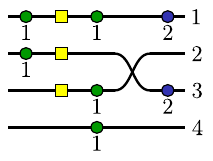}
\caption{Logical $S$ gate from swapping genons on qubits 2 and 3.}
\end{subfigure}
\caption{The [[4,1,2]] protocol implementing single qubit Clifford gates by braiding (swapping) genons.}
\label{fig:412-protocol}
\end{figure*}

\begin{figure*}[]
\centering
\includegraphics[scale=0.5]{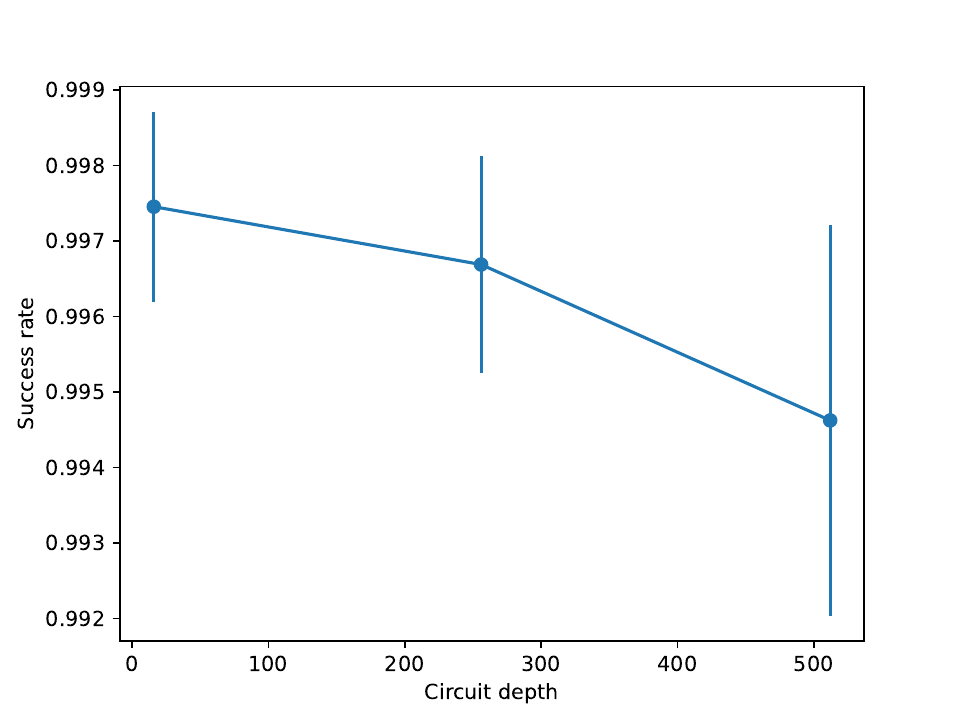}
\caption{The survival probability from the randomized benchmarking protocol on the $[[4,1,2]]$ code, as realized on the H1-1 machine. Circuit depths 16 and 256 were ran with 1600 shots each and circuit depth 512 was ran with 800 shots. We note this graph seems to show a quadratic decay as opposed to the linear decay more commonly seen by randomized benchmarking.}
\label{fig:412-results}
\end{figure*}
As a demonstration of genon braiding,
we ran \textit{logical randomized benchmarking}~\cite{Magesan2011, combes2017logical, mayer2024benchmarking},
on the $[[4,1,2]]$ genon code 
using the circuits shown in Fig.~\ref{fig:412-protocol}. 
See Appendix \ref{sec:qasm-412} for example QASM.
The protocol proceeds in 3 steps:

(Step 1)
We begin by preparing the logical $\ket{0}$ state using
the circuit in \Fref{fig:412-protocol}(ii).

(Step 2)
We apply a random sequence of $N$ logical Clifford gates.
There are 192 Clifford gates, or 24 up to phase, to choose from.
This group is generated by the $H$ and $S$ gates, 
which at the physical level is 
realized through concatenation of 
the circuits in Fig.~\ref{fig:412-protocol}(iii) and (iv).
We also use the redundant Clifford generators $X$ and $Z$,
coming from the logical Pauli operators of the $[[4,1,2]]$ code.
Each resulting concatenation is compiled for the H1-1 device 
into at most four single qubit Clifford gates;
qubit permutations are implemented via 
relabelling.

(Step 3)
We apply the inverse of the encoding circuit 
in \Fref{fig:412-protocol}(i). This gives
syndrome qubits 1,2 and 3 and (unencoded) logical qubit 4. 
We apply the inverse of the Clifford operation in (Step 2)
on the qubit 4, and then measure all
qubits in the $\bra{0,1}$ basis.

We treat this as an \emph{error detecting} code,
so the result of the measurement is discarded if
any of the syndrome bits are non-zero. 
Otherwise we record \emph{survival} if the logical bit is zero.
We ran three different random sequences of lengths $N = 16,\ 256$ and 512.
Using Step 1-3 above we sample 16 circuits 
for each of $N=16$ and $N=256$, and for $N=512$ we sampled 8 circuits.
Each circuit is then executed for 100 shots.
The discard counts were 28/1600, 90/1600, and 56/800 respectively.
The survival probability 
is plotted in Fig.~\ref{fig:412-results}.

In general, it can be hard to extract the logical fidelity from the randomized benchmarking protocol without interleaving QEC cycles ~\cite{mayer2024benchmarking}. 
From the randomized benchmarking protocol, we would expect to see a linear decay in the survival probability resulting from accumulation of errors \cite{Magesan2011}. 
However, we observe that the curve seen in Fig.~\ref{fig:412-results}, matches a quadratic fit instead of a linear one. 
Further investigation is needed to conclude whether such a logical randomized benchmarking experiment, without interleaving QEC cycles, is sufficient to extract a reliable logical error rate and is outside the scope of this work. 
Here we use the randomize benchmarking protocol as a proof of practice that genon braiding on the $[[4,1,2]]$ can be used to realize the full Clifford group and easily implemented on the H1-1 device.

\subsection{The [[8,2,2]] protocol}

In this experiment we benchmark a logical $CX$ gate
on the $[[8,2,2]]$ toric code.
This code is the double cover of the $[[4,1,2]]$ base code,
and the logical $CX$ is the lift of the 
logical $S$ gate on the base code
given in \Fref{fig:412-protocol} (iv),
using Theorem \ref{th:lift}.
The fibers over the base qubits $1,2,3,4$ are 
$\{1,5\},\{2,6\},\{3,7\},\{4,8\}$ respectively,
see \Fref{fig:822-protocol}.

To benchmark this $CX$ gate we run $4+8=12$ circuits comprised of
the following:
\begin{itemize}
\item State preparation and measurement (SPAM) circuits, for each of logical
$\ket{00}, \ket{11}, \ket{++},$ and $\ket{--}.$
\item The action of the logical $CX$ on the two qubit computational logical states 
$\ket{00}, \ket{01}, \ket{10},$ and $\ket{11}$ as well as 
the two qubit phase logical states 
$\ket{++}, \ket{+-}, \ket{-+},$ and $\ket{--}$.
\end{itemize}
At the end of the circuit we measure the
qubits in the 
$\ket{0,1}^{\tensor 8}$ basis for the computational logical states,
or the 
$\ket{+,-}^{\tensor 8}$ basis for the phase logical states.
As this is an error detecting code, 
if the measurement fails to satisfy the $Z$-checks, or $X$-checks
of the $[[8,2,2]]$ code respectively, the result is discarded. 
Otherwise, an error is recorded when the measurement fails to lie within
the $X$ or $Z$-type stabilizer group of the $[[8,2,2]]$ code respectively.

Each of these circuits is run for 5000 shots. 
We also simulate these circuits for 50000 shots, with the
error count then divided by 10.
These results are tabulated in Table~\ref{Tab:822_states}
with the logical fidelities calculated in Table~\ref{Tab:822_Fids}.
The (simulated) overall accept probability was $96\pm 1\%$. 
From the experiments, it appears that the $X$ basis is a more robust to noise than the $Z$, having no logical errors in SPAM or the $CX$ experiments. We attribute the difference between the two basis to the difference circuit depth of the encoding circuits seen in Fig.~\ref{fig:822-protocol}. We note that these encoding circuits were not tested for fault-tolerant properties, meaning it was not confirmed that higher weight errors do not propagate through. Further analysis is needed to construct shallower, fault-tolerant encoding circuits and an general encoding protocol, but we leave this to future work.

\begin{center}
\begin{table}
\begin{tabular}{|c|c|c|c|}
\hline
logical &\nsp logical \nsp& \nsp experimental\nsp    &\nsp simulated\nsp \\
operation &  state      &  errors         & errors  \\
\hline
$I$ &  $\ket{00}$           &   3    &  4.8  \\
\hline
$I$ &  $\ket{11}$           &   4    &  6.9  \\
\hline
\hline
$CX_{1,0}$ &  $\ket{00}$     &   0    &  4.3  \\
\hline
$CX_{1,0}$ &  $\ket{01}$     &   2    &  4.7  \\
\hline
$CX_{1,0}$ &  $\ket{10}$     &   2    &  3.6  \\
\hline
$CX_{1,0}$ &  $\ket{11}$     &   3    &  4.6  \\
\hline
\end{tabular}
\begin{tabular}{|c|c|c|c|}
\hline
logical &\nsp logical \nsp& \nsp experimental\nsp    &\nsp simulated\nsp \\
operation &  state      &  errors         & errors  \\
\hline
$I$ &  $\ket{++}$           &   0    &  0.5  \\
\hline
$I$ &  $\ket{--}$           &   0    &  0.6  \\
\hline
\hline
$CX_{1,0}$ &  $\ket{++}$     &   0    &  0.7  \\
\hline
$CX_{1,0}$ &  $\ket{+-}$     &   0    &  0.5  \\
\hline
$CX_{1,0}$ &  $\ket{-+}$     &   0    &  0.3  \\
\hline
$CX_{1,0}$ &  $\ket{--}$     &   0    &  1.0  \\
\hline
\end{tabular}
\caption{The number of logical errors found for the $[[8,2,2]]$ code simulations and experiments. Here the logical operation $I$ is meant to imply the state preparation and measurement errors seen while preparing the individual two-qubit logical states. A series of experiments were also performed implementing the logical $CX$ gate between the two logical qubits contained in the $[[8,2,2]]$ code block.}
\label{Tab:822_states}
\end{table}
\end{center}

\begin{table}
\begin{center}
\begin{tabular}{|c | c | c | c |}
    \hhline{~|---|} \multicolumn{1}{c |}{}& X basis & Z basis & Average \\
    \hhline{|----|}
 $SPAM_{exp}$ & $1.0000_{-2}^{+0}$ & $0.9993_{-3}^{+3}$ & $0.9996_{-2}^{+2}$   \\
 \hhline{|----|}
 $SPAM_{sim}$ & $0.9999_{-3}^{+3}$ & $0.9988_{-1}^{+1}$ & $0.9993_{-1}^{+1}$   \\
 \hhline{|----|}
 $CX_{exp}$ & $1.0000_{-2}^{+0}$ & $0.9996_{-2}^{+2}$ & $0.9998_{-1}^{+1}$   \\
 \hhline{|----|}
 $CX_{sim}$ & $0.99987_{-4}^{+4}$ & $0.9991_{-1}^{+1}$ & $0.9995_{-3}^{+3}$   \\
\hhline{|----|}
\end{tabular}
\caption{The logical fidelities for state preparation and measurement (SPAM) as well as the $CX$ implementation for the $[[8,2,2]]$ code. We note that $X$ basis appears to perform better than $Z$ basis, which we attribute to the depth of the encoding circuits used. }
\label{Tab:822_Fids}
\end{center}
\end{table}

\begin{figure*}[]
\centering
\begin{subfigure}[t]{0.30\textwidth}
\igr{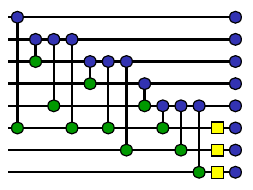}
\caption{Preparing logical $\ket{00}$.}
\end{subfigure}
\begin{subfigure}[t]{0.30\textwidth}
\igr{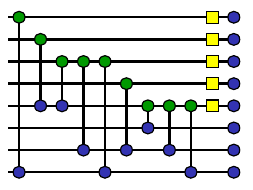}
\caption{Preparing logical $\ket{++}$.}
\end{subfigure}
\begin{subfigure}[t]{0.30\textwidth}
\igr{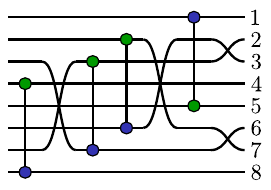}
\caption{Logical $CX_{1,0}$ gate, with input qubits numbered $1-8$.
}
\end{subfigure}
\caption{The [[8,2,2]] Dehn twist protocol for benchmarking a logical $CX$ gate. We note these encoding circuit have not been tested for fault-tolerant properties, and may lead to higher-weight physical errors to propagate to the logical state.}
\label{fig:822-protocol}
\end{figure*}

\subsection{The [[10,2,3]] protocol}

Here we benchmark the two qubit logical gate:
$$
g=CX_{0,1}\cdot SWAP 
$$
This is found as the lift of the order three (up to phase)
transversal Clifford gate $SH$ on the $[[5,1,3]]$ base code, using 
Theorem \ref{th:lift}.
See \Fref{fig:1023-protocol}.
We follow a similar benchmarking procedure as for the $[[8,2,2]]$
protocol above with 12 different circuits, 
except that instead of discarding measurements that
present a non-zero syndrome we now infer a logical correction operator
from the syndrome data 
using a \emph{decoder algorithm}.
This decoder accepts one of $2^4=16$ possible syndrome bits and outputs
the most likely of the $2^2=4$ logical correction operators.
We pre-calculate all of these using simulated data on $10^5$ shots,
and generate a lookup table.
Note that we build a separate lookup table for each of the 16 circuits,
in this way the decoder is aware of the specific noise model of that circuit.
This improves the performance of the benchmark substantially.
The shots where the decoder fails to give the correct logical
operation are then recorded as errors.

We tabulate experimental results in Table ~\ref{Tab:1023_states} and fidelities of these operations in Table ~\ref{Tab:1023_Fids}.
Each circuit is run for 2000 shots.
We also simulate each circuit for 
$2\times 10^5$ shots and then normalize to $2000$ shots.
Similar to the $[[8,2,2]]$ results, we see a difference between the $X$ basis and $Z$ basis, with the $X$ performing a bit better. We again attribute this to the circuit depth of the encoding circuits as seen in Fig.~\ref{fig:1023-protocol}. 

\begin{table}
\begin{center}
\begin{tabular}{|c|c|c|c|}
\hline
logical &\nsp logical \nsp& \nsp experimental\nsp    &\nsp simulated\nsp \\
operation &  state      &  errors         & errors  \\
\hline
$I$ &  $\ket{00}$           &   4    &  6.68  \\
\hline
$I$ &  $\ket{11}$           &   2    &  7.14  \\
\hline
\hline
$g$ &  $\ket{00}$     &   6    &  14.18  \\
\hline
$g$ &  $\ket{01}$     &   7    &  13.84  \\
\hline
$g$ &  $\ket{10}$     &   4    &  14.29  \\
\hline
$g$ &  $\ket{11}$     &   2    &  14.97  \\
\hline
\end{tabular}
\begin{tabular}{|c|c|c|c|}
\hline
logical &\nsp logical \nsp& \nsp experimental\nsp    &\nsp simulated\nsp \\
operation &  state      &  errors         & errors  \\
\hline
$I$ &  $\ket{++}$           &   1    &  3.05  \\
\hline
$I$ &  $\ket{--}$           &   3    &  3.16  \\
\hline
\hline
$g$ &  $\ket{++}$     &   8    &  12.30  \\
\hline
$g$ &  $\ket{+-}$     &   2    &  12.65  \\
\hline
$g$ &  $\ket{-+}$     &   3    &  12.58  \\
\hline
$g$ &  $\ket{--}$     &   4    &  13.22  \\
\hline
\end{tabular}
\caption{The number of logical errors found for the $[[10,2,3]]$ code simulations and experiments. Here the logical operation $I$ is meant to imply the state preparation and measurement errors seen while preparing the individual two-qubit logical states. A series of experiments were also performed implementing the logical $g$ gate between the two logical qubits contained in the $[[10,2,3]]$ code block.}
\label{Tab:1023_states}
\end{center}
\end{table}

\begin{table}
\begin{center}
\begin{tabular}{|c | c | c | c |}
    \hhline{~|---|} \multicolumn{1}{c |}{}& X basis & Z basis & Average \\
    \hhline{|----|}
 $SPAM_{exp}$ & $0.9990_{-7}^{+7}$ & $0.9985_{-8}^{+8}$ & $0.9987_{-7}^{+7}$   \\
 \hhline{|----|}
 $SPAM_{sim}$ & $0.99844_{-8}^{+8}$ & $0.9965_{-1}^{+1}$ & $0.9974_{-1}^{+1}$   \\
 \hhline{|----|}
 $g_{exp}$ & $0.997_{-1}^{+1}$ & $0.997_{-1}^{+1}$ & $0.997_{-1}^{+1}$   \\
 \hhline{|----|}
 $g_{sim}$ & $0.9936_{-1}^{+1}$ & $0.9928_{-1}^{+1}$ & $0.9932_{-1}^{+1}$   \\
\hhline{|----|}
\end{tabular}
\caption{The logical fidelities for state preparation and measurement (SPAM) as well as the $g$ implementation. We note that $X$ basis appears to perform better than $Z$ basis, which we attribute to the depth of the encoding circuits used. }
\label{Tab:1023_Fids}
\end{center}
\end{table}

\begin{figure*}[]
\centering
\begin{subfigure}[t]{0.40\textwidth}
\igr{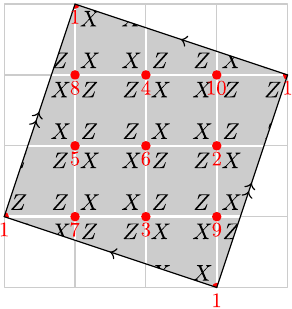}
\caption{The qubit indexes}
\end{subfigure}
\begin{subfigure}[t]{0.40\textwidth}
\igr{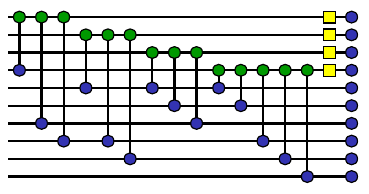}
\caption{Preparing logical $\ket{00}$.}
\end{subfigure}
\begin{subfigure}[t]{0.40\textwidth}
\igr{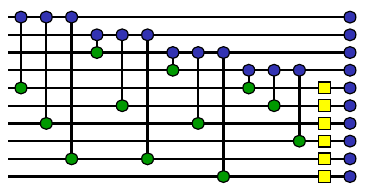}
\caption{Preparing logical $\ket{++}$.}
\end{subfigure}
\begin{subfigure}[t]{0.40\textwidth} \igr{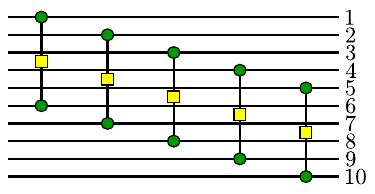} \caption{A logical $CZ$ gate.  } \end{subfigure} 
\begin{subfigure}[t]{0.40\textwidth} \igr{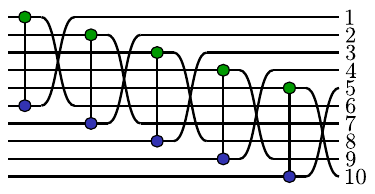} \caption{A logical $CX_{0,1}\cdot SWAP$ gate.  } \end{subfigure}
\caption{The [[10,2,3]] protocol is derived from the $[[5,1,3]]$ base code.}
\label{fig:1023-protocol}
\end{figure*}

\section{Conclusion}\label{sec:conc} 

As the field of quantum error correction advances, significant
progress has been made in the design of 
quantum low-density parity check codes
with favorable properties. However,
many of these codes are limited by their scarcity of
fault-tolerant logical operations. In this work we seek
to \emph{co-design} quantum codes that have both reasonable
scaling properties as well as giving fault-tolerant logicals
beyond the Pauli group.

Our study explores the use of covering spaces, a concept
that underlies mathematical fields including
Galois theory, number theory, and algebraic topology.
Specifically, we focus on double covers within the realms
of symplectic geometry, quantum codes, Riemann surfaces,
and topological quantum codes. This multidisciplinary
approach underscores a broader theoretical idea:
two-dimensional topologically ordered systems should
exhibit a correspondence between domain walls and covering
spaces, particularly in the context of abelian domain walls.

A significant contribution of our work is the explicit
protocol we develop for braiding genons and performing
Dehn twists. This protocol leverages qubit permutations
and constant depth Clifford circuits, which are efficiently
realizable on quantum computers with high connectivity,
such as Quantinuum’s H1-1. The practical implementation
of these gates results in robust logical fidelity, showcasing
the experimental viability of our approach.

Non-Clifford gates are essential for achieving universal fault-toleran
quantum computation. While Clifford gates alone form
a useful set for many quantum operations, they are insufficient
for universal quantum computing. Our construction lays
the groundwork for integrating non-Clifford gates into
the topological code framework, a critical step for universal
fault-tolerant computation. Further research is required
to fully develop and implement non-Clifford gates within
these codes. Nonetheless, the methods and constructions
found in this work appear promising and compatible with
existing approaches, suggesting a viable pathway toward
their realization.

Looking ahead, our findings suggest promising directions for
further exploration. The correspondence between domain
walls and covering spaces observed in two-dimensional
topological systems could extend to three-dimensional
systems. Such systems might exhibit defects whose statistics
enable the generation of non-Clifford gates, pushing
the boundaries of fault-tolerant quantum computation.
Drawing inspiration from number theory, algebraic geometry,
and related fields, we envision the development of more
sophisticated quantum codes and fault-tolerant operations
that could revolutionize quantum computing.

In conclusion, our work lays the groundwork for a unified
theory that bridges diverse mathematical disciplines
and advances the design of quantum error-correcting codes.
By integrating twists, defects, and higher-dimensional
topological structures, we open new pathways toward achieving
versatile and scalable quantum computation with enhanced
fault tolerance.

{\bf Acknowledgements.}
The authors thank Thomas Scruby, Michael Vasmer, Karl Mayer, 
 Shival Dasu, Dan Gresh, and Pablo Andres-Martinez
 for useful conversations and feedback on this work.

\bibliography{refs}{}
\bibliographystyle{abbrv}
\appendix
\section{Fault-Tolerance of Fiber Transversal Gates}\label{sec:proof}
Given a base code with $m\times 2n$ check matrix $\Parity = \bigl( \Parity_X\ \Parity_Z \bigr)$, the doubled code $\Double(\Parity)$ has a $2m\times 4n$ parity check matrix of the form \ref{eq:dh}, repeated here for convenient reference:
\begin{align}
\Double(\Parity) := 
\begin{pmatrix}
\Parity_X & \Parity_Z & 0 & 0 \\
0 & 0 & \Parity_Z & \Parity_X 
\end{pmatrix}.
\end{align}
\begin{theorem}\label{th:ft}
Given a fault-tolerant gate on quantum code $C$ with parameters $[[n,k,d]]$, the lifted gate on the doubled code $\Double(C)$ is also fault-tolerant to at least distance $d$.
\end{theorem}

\begin{proof}
This proof will show a correspondence between the syndromes of faults on the base and lifted codes. If the gate on the base code tolerates the original fault up to weight $t=\lfloor \frac{d-1}{2} \rfloor$, the lifted gate on the doubled code tolerates the transformed fault. \\
If a gate on the base code is supported on qubit $i$ on the base code, in the doubled code, the lifted gate has a two-qubit gate supported on qubits $i$ and $i+n$.\\
Given a fault $f\in \Field_2^{n}\oplus\Field_2^n$ written in block column form: 
\begin{align}
f = \begin{pmatrix} f_X \\ f_Z \end{pmatrix},
\end{align}
we define the \emph{syndrome} of the fault, $S_f \in \Field_2^m$: 
\begin{align}\label{syn:base}
S_f := \Parity \Omega_{n} f = \bigl( \Parity_X\ \Parity_Z \bigr) \Omega_{n} \begin{pmatrix} f_X \\ f_Z \end{pmatrix} = \Parity_X f_Z + \Parity_Z f_X.
\end{align}
For the remainder of the proof, we will speak about some single, constant fault $f$. \\
In the doubled code, the syndrome of a fault $f^\prime\in \Field_2^{2n}\oplus\Field_2^{2n}$ is $\mathbf{S}_{f^\prime} \in \Field_2^{2m}$.
\begin{align}
\mathbf{S}_{f^\prime} =  \begin{pmatrix}
\Parity_X & \Parity_Z & 0 & 0 \\
0 & 0 & \Parity_Z & \Parity_X 
\end{pmatrix} \Omega_{2n} \begin{pmatrix} f^\prime_X \\ f^\prime_Z \end{pmatrix} = (\Parity_X \Parity_Z)f^\prime_Z \oplus (\Parity_Z \Parity_X)f^\prime_X.
\end{align}

Since the doubled code is a CSS code, we may break the syndrome into $X$ and $Z$ components. The doubled parity check matrix also has equal size $X$ and $Z$ components, so the syndrome may be represented: 
\begin{align}
    \mathbf{S}_{f^\prime}= \mathbf{S}_{f^\prime}^X \oplus \mathbf{S}_{f^\prime}^Z && \mathbf{S}_{f^\prime}^X , \mathbf{S}_{f^\prime}^Z \in \Field_2^m.
\end{align}

These parts of the syndromes are calculated:
\begin{align}\label{syn:double}
    \mathbf{S}_{f^\prime}^X = (\Parity_X \Parity_Z)f^\prime_Z && \mathbf{S}_{f^\prime}^Z = (\Parity_Z \Parity_X)f^\prime_X .
\end{align}
We observe that, using ~\ref{syn:double} and ~\ref{syn:base}, if $f^\prime_Z = \begin{pmatrix} f_Z \\ f_X \end{pmatrix}$, 
then $\mathbf{S}_{f^\prime}^X = S_f$. Note that $w\begin{pmatrix} f_Z \\ f_X \end{pmatrix} = w\begin{pmatrix} f_X \\ f_Z \end{pmatrix}$. 
Therefore, if there is a decoder on the base such that it corrects all $\{f|w(f)\leq t\}$, there is also a decoder on the lifted code which corrects all $f^\prime_Z$ of the form $\begin{pmatrix} f_Z \\ f_X \end{pmatrix}$.\\

In particular, a $t$-fault-tolerant base code corrects $Y$-type Pauli faults of weight less than $t$. $Y$-type errors of this form satisfy $f_Z = f_X$ and $w(f)=w(f_Z)=w(f_X)\leq t$. For clarity, we will represent these symmetric faults as $\begin{pmatrix} f_Y \\ f_Y \end{pmatrix}$. This implies that the doubled code can correct faults of the form $\begin{pmatrix} f_Y \\ f_Y \end{pmatrix}$ where $w(f_Y)\leq t$, or in Pauli notation $Z_i Z_{i+n}$. These are exactly the weight two $Z$-type faults resulting from lifted gates on the doubled code. Also, a $t$-fault-tolerant base code can correct $X$ and $Z$ type faults with block forms $\begin{pmatrix} f_X \\ 0 \end{pmatrix}$ and $\begin{pmatrix} 0 \\ f_Z \end{pmatrix}$ respectively. Thus, the doubled code can correct all faults of the form $\begin{pmatrix} f_Y + f_X \\f_Y + f_Z \end{pmatrix}$ satisfying $w(f_X) + w(f_Y) + w(f_Z) \leq t$, which spans the space of up to $t$ faults on lifted gates. 

Since the doubled code is a CSS code, the proof for $X$-type faults is the same, but with the roles of $\Parity_X$ and $\Parity_Z$ reversed. \\
\end{proof}

\section{Example QASM}\label{sec:qasm-412}

This is QASM source for an example of the $[[4,1,2]]$ randomized benchmarking protocol
of length $N=2$.
We show the qubit permutations as {\tt P(...)} operations
as well as the resulting {\tt labels} in comments.

\begin{verbatim}
OPENQASM 2.0;
include "hqslib1.inc";

qreg q[4];
creg m[4];

reset q;

// prepare logical |0> state
h q[1]; h q[0]; cy q[0], q[1]; h q[2];
cy q[1], q[2]; cx q[0], q[3]; h q[0]; x q[0]; z q[0];
barrier q;

// First logical Clifford
x q[2]; x q[0];
// P(0, 2, 1, 3)
// labels = [0, 2, 1, 3]
s q[3]; s q[1]; h q[1]; h q[2]; s q[2]; s q[0]; h q[0]; s q[0];
// P(0, 3, 2, 1)
// labels = [0, 3, 1, 2]
x q[1]; x q[0]; h q[2]; h q[1]; h q[3]; h q[0]; x q[1]; z q[3]; x q[1]; x q[0];
// P(0, 2, 1, 3)
// labels = [0, 1, 3, 2]
s q[2]; s q[3]; h q[3]; h q[1]; s q[1]; s q[0]; h q[0]; s q[0];
// P(0, 3, 2, 1)
// labels = [0, 2, 3, 1]
x q[3]; x q[0]; h q[1]; h q[3]; h q[2]; h q[0]; x q[3]; z q[2]; x q[2]; z q[0];
barrier q;

// Second logical Clifford
x q[3]; x q[0];
// P(0, 2, 1, 3)
// labels = [0, 3, 2, 1]
s q[1]; s q[2]; h q[2]; h q[3]; s q[3]; s q[0]; h q[0]; s q[0];
// P(0, 3, 2, 1)
// labels = [0, 1, 2, 3]
x q[2]; x q[0]; h q[3]; h q[2]; h q[1]; 
h q[0]; x q[2]; z q[1]; x q[2]; z q[1]; x q[1]; z q[0];
barrier q;

// decode
cy q[0], q[1]; cz q[0], q[2]; h q[0];
cy q[1], q[2]; cz q[1], q[3]; h q[1];
cz q[2], q[0]; cy q[2], q[3]; h q[2];
sdg q[3]; h q[3]; s q[3];

// inverse of logical Clifford
x q[3]; z q[3]; h q[3]; sdg q[3];
x q[3]; h q[3]; sdg q[3]; h q[3];
sdg q[3];

measure q -> m;
// final qubit order: [0, 1, 2, 3]
\end{verbatim}

For the $[[8,2,2]]$ protocol, 
we show example QASM preparing the $\ket{0,1}$ state and
applying the fault-tolerant Clifford gate $CX_{0,1}\cdot SWAP$.

\begin{verbatim}
OPENQASM 2.0;
include "hqslib1.inc";

qreg q[8];
creg m[8];

reset q;

h q[7]; h q[6]; h q[5];
cx q[7], q[4]; cx q[6], q[4]; cx q[5], q[4];
cx q[4], q[3]; cx q[6], q[2]; cx q[5], q[2];
cx q[3], q[2]; cx q[5], q[1]; cx q[4], q[1];
cx q[2], q[1]; cx q[5], q[0];
x q[4]; x q[1];
barrier q;

// P(0, 2, 1, 3, 4, 6, 5, 7)
// labels = [0, 2, 1, 3, 4, 6, 5, 7]
cx q[4], q[0];
// P(0, 5, 2, 3, 4, 1, 6, 7)
// labels = [0, 6, 1, 3, 4, 2, 5, 7]
cx q[6], q[2];
cx q[1], q[5];
// P(0, 1, 6, 3, 4, 5, 2, 7)
// labels = [0, 6, 5, 3, 4, 2, 1, 7]
cx q[3], q[7];
barrier q;

x q[4]; x q[6]; x q[7]; x q[0];
measure q -> m;

// final qubit order: [0, 6, 5, 3, 4, 2, 1, 7]
\end{verbatim}

Here is a circuit for the $[[10,2,3]]$ protocol, acting on logical $\ket{+-}$:

\begin{verbatim}
OPENQASM 2.0;
include "hqslib1.inc";

qreg q[10];
creg m[10];

reset q;

h q[9]; h q[8]; h q[7]; h q[6]; h q[5]; h q[4];
cx q[7], q[3]; cx q[5], q[3]; cx q[4], q[3];
cx q[9], q[2]; cx q[6], q[2]; cx q[3], q[2];
cx q[8], q[1]; cx q[5], q[1]; cx q[2], q[1];
cx q[8], q[0]; cx q[6], q[0]; cx q[4], q[0];
barrier q;

z q[5]; z q[4]; z q[1];
barrier q;

// P(0, 1, 2, 3, 9, 5, 6, 7, 8, 4)
// labels = [0, 1, 2, 3, 9, 5, 6, 7, 8, 4]
cx q[9], q[4];
// P(0, 1, 2, 8, 4, 5, 6, 7, 3, 9)
// labels = [0, 1, 2, 8, 9, 5, 6, 7, 3, 4]
cx q[8], q[3];
// P(0, 1, 7, 3, 4, 5, 6, 2, 8, 9)
// labels = [0, 1, 7, 8, 9, 5, 6, 2, 3, 4]
cx q[7], q[2];
// P(0, 6, 2, 3, 4, 5, 1, 7, 8, 9)
// labels = [0, 6, 7, 8, 9, 5, 1, 2, 3, 4]
cx q[6], q[1];
// P(5, 1, 2, 3, 4, 0, 6, 7, 8, 9)
// labels = [5, 6, 7, 8, 9, 0, 1, 2, 3, 4]
cx q[5], q[0];
barrier q;

z q[9]; z q[8]; z q[7]; z q[6]; z q[5];
barrier q;

h q[4]; h q[3]; h q[2]; h q[1]; h q[0];
h q[9]; h q[8]; h q[7]; h q[6]; h q[5];
measure q -> m;

// final qubit order: [5, 6, 7, 8, 9, 0, 1, 2, 3, 4]
\end{verbatim}

\end{document}